\documentclass[12pt]{amsart}

\usepackage{fullpage, graphicx, amsmath, amssymb}

\newtheorem{thm}{Theorem}[section]
\newtheorem{prop}[thm]{Proposition}

\newtheorem{lemma}[thm]{Lemma}
\newtheorem{cor}[thm]{Corollary}
\newtheorem{definition}[thm]{Definition}
\theoremstyle{definition}

\def\One{\mathbb{I}}

\title{Diffeomorphisms of quantum fields}
\author{Dirk Kreimer and Karen Yeats}
\thanks{KY thanks the Kolleg Mathematik Physik Berlin (KMPB) for support during a visit in summer 2016.  KY is supported by an NSERC Discovery grant. DK thanks the IHES for hospitality September 2016.}
\begin{document}
\maketitle
\begin{abstract}
We study field diffeomorphisms $\phi(x)\to F(\phi(x))=a_0\phi(x)+a_1\phi^2(x)+\ldots=\sum_{j+0}^\infty a_j \phi^{j+1},
$ for free and interacting quantum fields $\Phi$. We find that the theory is invariant under such diffeomorphisms if and only if kinematic renormalization schemes are used.  
\end{abstract}
\section{Introduction}
\subsection{The problem and the results}
In \cite{KVdiff} Andrea Velenich and one of us investigated what happens if one applies a field diffeomorphism
\[
\phi(x)\to F(\phi(x))=a_0\phi(x)+a_1\phi^2(x)+\ldots=\sum_{j+0}^\infty a_j \phi^{j+1},
\]
to a free scalar quantum field $\Phi(x)$ with Lagrangian density
\[
L(\phi)=\frac{1}{2} \partial_\mu\phi(x)\partial^\mu\phi(x)-\frac{m^2}{2}\phi^2(x).
\]
We set $a_0=1$ (and the diffeomorphism is tangent to the identity, so no constant term) in the following without loss of generality. The question to study is how, in terms of these seemingly interacting fields,
one recognizes the underlying  free field theory. No recourse to formal manipulations of a path integral or the path integral measure 
was made in \cite{KVdiff} nor is it made here. 

Instead, in the context of kinematic renormalization schemes, it was shown for the massive theory that 
interacting tree-level amplitudes vanish, through explicit computations summing all amplitudes up to six external legs. 
The vanishing reveals itself only in the sum of all tree amplitudes with a given number of external legs and is based on non-trivial cross cancellations. 

In that first paper we could not provide an all orders proof of the vanishing of the tree-level amplitudes. This is the crucial requirement to understand the situation in general though: the vanishing of loop amplitudes follows from the vanishing of tree-amplitudes and analytic properties of amplitudes in the context of those renormalization schemes which subtract at a renormalization point given by kinematic conditions on the amplitude.

With loops, the same was established at first loop order for such kinematic renormalization schemes.\footnote{In \cite{Gervais} a similar result was obtained.
A formal use of a Jacobian of the field diffeomorphism  in the path integral leads to erroneous results by terms which would vanish in kinematic renormalization. There, a solution was proposed modifying the path integral formalism.} For the massless theory, the vanishing of all interacting tree- and loop-amplitudes was shown on analysing the structure of the $S$-matrix.

In this document we will prove the vanishing of interacting amplitudes after diffeomorphism  for all $n>2$ and for  all loop orders, where $n$ is the number of external legs for the amplitude under consideration, with $n=2$ describing propagation, and $n>2$ interaction.

The first step is to prove it at tree-level.  This involves two steps, a reduction of the problem to a purely combinatorial identity and a proof of this identity involving some non-trivial manipulations of Bell polynomials.  In order to achieve this we will first need a digression into Bell polynomials and Bell polynomial identities.  It shouldn't be surprising that Bell polynomials are important since Bell polynomials can be used to describe compositions of series, see further comments at the beginning of Section~\ref{sec bell}.

Once the tree-level result is proved we will extend it inductively to all loop orders using that loop amplitudes are formed from tree amplitudes with off-shell legs identified between trees, or equivalently from Cutkosky rules and the optical theorem.  

We verify  that taking the sum over all possible left hand and right hand sides  --- for a gluing of loop amplitudes from two tree-amplitudes or vice versa from cutting a loop amplitude for a Cutkosky cut --- with each diagram on each side weighted by its symmetry factor will give the correct symmetry factors for the full diagrams.  

Once that is in hand the sums on both sides of the cut are themselves vanishing amplitudes inductively since the cut edges are on-shell and so Cutkosky tells us that the imaginary part of the whole amplitude vanishes.  Then the optical theorem gives that the amplitude itself vanishes proving the main result. Equivalently, upon gluing Feynman rules of the diffeomorphed theory reduce amplitudes to tadpole amplitudes which vanish in kinematic renormalization. 

Our final result is then in accordance with the result in \cite{Flume}
which was obtained in coordinate space renormalization: invariance under point transformation for a theory renormalized by local subtractions.

We will then continue and study field diffeomorphisms tangent to the identity for an interacting field theory.
Again, we show that the structure of the newly generated vertices and the demands of $S$-matrix theory suffice to conclude the diffeomorphisms invariance of Wightman functions --- but only in the context of kinematic renormalization. Remarkably, the Jacobian of the field diffeomorphisms plays no role in this proof. We conclude with some considerations on the equivalence class 
defined by field diffeomorphisms tangent to the identity and exhibit consequences for the adiabatic limit in the context of Haag's theorem.
\subsection{Set-up and previous work}\label{subsec prev}
It has long been known that loop contributions to quantum S-matrix elements can be obtained from tree-level amplitudes using unitarity methods based on the optical theorem and dispersion relations. Indeed, $d$-dimensional unitarity methods allow us to compute S-matrix elements without the need of an underlying Lagrangian and represent an alternative to the usual quantization prescriptions based on path integrals or canonical quantization. 

We consider field diffeomorphisms of a free field theory which generate a seemingly interacting field theory. 
This is an old albeit somewhat controversial topic in the literature,
see \cite{KVdiff} and references there. 

As in \cite{KVdiff} we address it here from a minimalistic approach ignoring any path-integral heuristics.

As any interacting field theory, an interacting  field theory whose interactions originate from field diffeomorphisms
of a free field theory alone has a perturbative expansion which is governed by a corresponding tower of Hopf algebras \cite{Suj}. 
It starts from the core Hopf algebra for which only one-loop graphs are primitive:
\begin{equation}
\Delta \Gamma=\Gamma\otimes\One+\One\otimes\Gamma +\sum_{\cup_i\gamma_i=\gamma\subset\Gamma}\gamma\otimes\Gamma/\gamma,
\end{equation}
 and ends with a Hopf algebra for which any one-particle irreducible (1PI) graph is primitive:
\begin{equation}
\Delta \Gamma=\Gamma\otimes\One+\One\otimes\Gamma.
\end{equation}
Here, subgraphs $\gamma_i$ are 1PI. Intermediate between these two Hopf algebras are those for which graphs of a prescribed superficial degree of divergence contribute
in the coproduct.

All these Hopf algebras allow for maximal co-ideals. The core Hopf algebra has a maximal  ideal which relates to the celebrated BCFW relations: if the latter relations hold, the Feynman rules are well defined on the quotient of the core Hopf algebra by  this maximal ideal \cite{Suj}.

In the Hopf algebra of Feynman diagrams such  Hopf co-ideals are known to encode the symmetries of a field  theory \cite{Suj}. Such co-ideals enforce relations among the n-point 1-particle irreducible Green functions $\Gamma^{(n)}_{1PI}$ or among the connected Green functions $\Gamma^{(n)}_c$, which generically are of the form:
\begin{equation} \label{decore}
\Gamma^{(n)}_{1PI} = \Gamma^{(j)}_{1PI} \frac{1}{\Gamma^{(2)}_{1PI}} \Gamma^{(k)}_{1PI} \qquad \forall \; j,k>2 \; ; \; j+k = n+2.
\end{equation}
with relations which characterize a Hopf ideal \cite{Suj}.
This is equivalent to the identity familiar in gauge theory \cite{anatomy} and quantum gravity \cite{Kgravity}:
\[
\frac{\Gamma^{(n+1)}_{1PI}}{\Gamma^{(n)}_{1PI}}
=
\frac{\Gamma^{(n)}_{1PI}}{\Gamma^{(n-1)}_{1PI}},\, n\geq 3.
\]

For gravity, the core Hopf algebra is involved as gravity Feynman rules are necessarily such that the power-counting for 
vertices (involving two derivatives at a vertex) compensates the power-counting for internal edges \cite{Kgravity}. The same power-counting is generated through our field diffeomorphisms through the 
\[
\frac{1}{2}\partial_\mu F(\phi) \partial^\mu F(\phi)
\]
term, motivating future study. 

The corresponding Hopf ideal is related to the diffeomorphism invariance of the  theory:
we prove that the connected amplitudes vanish which leads to the relations Eq.(\ref{decore}) above, for example for the four-point function we get
\[
0=\Gamma^{4}_c=\Gamma^{(4)}_{1PI}+\Gamma^{(3)}_{1PI}\frac{1}{\Gamma^{(2)}_{1PI}}\Gamma^{(3)}_{1PI}
\]
\begin{figure}[!h]
\includegraphics[width=12cm]{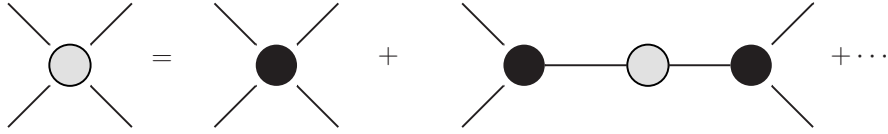}
\caption{The identity $\Gamma^{(4)}_c=\Gamma^{(4)}_{1PI}+\Gamma^{(3)}_{1PI}\frac{1}{\Gamma^{(2)}_{1PI}}\Gamma^{(3)}_{1PI}$.}
\label{conn1PI}
\end{figure}
Note that the 1PI  two-point function $\Gamma^{(2)}_{1PI}=1-\Sigma$ is never vanishing: a free field theory provides the lowest order in the perturbation expansion of a field theory, $\Gamma^{(2)}_{1PI}\not=0$ even for vanishing interactions.  
Hence if
$\Gamma^{(3)}_{1PI}=0$ in Eq.(\ref{decore}) we conclude $\Gamma^{(n)}_c=0,n\geq 3$.

To start, let us consider real scalar fields then defined on a 4-dimensional Minkowski space-time 
$\phi\equiv\phi(x):\mathbb{R}^{1,3} \rightarrow \mathbb{R}$ and field diffeomorphisms $F(\phi)$ specified by choosing a set of real coefficients $\{ a_k \}_{k \in \mathbb{N}}$ which do not depend on the space-time coordinates:
\begin{equation} \label{defF}
F(\phi) = \sum_{k=0}^{\infty} a_k\phi^{k+1} = \phi + a_1 \phi^2 + a_2 \phi^3 + \ldots \qquad (\mathrm{with} \; a_0=1).
\end{equation}
These transformations are often called point transformations. They preserve Lagrange's equations, they are a subset of the canonical transformations \cite{castro}, and in the quantum formalism they become unitary transformations of the Hamiltonian \cite{nakai}. 

The two field theories which we will consider are derived from the free massless and the massive scalar field theories, 
with Lagrangian densities ${L}(\phi)$ and with $F$ defined as in (\ref{defF}):
\begin{eqnarray}
\label{L1} & L(\phi)=\frac{1}{2}\partial_\mu\phi\partial^\mu \phi \to L_F(\phi) = \frac{1}{2} \partial_{\mu} F(\phi) \, \partial^{\mu} F(\phi), \\
\label{L2} & L(\phi)=\frac{1}{2}\partial_\mu\phi\partial^\mu \phi-\frac{m^2}{2}\phi^2 \to L_F(\phi) = \frac{1}{2} \partial_{\mu} F(\phi) \, \partial^{\mu} F(\phi) -
 \frac{m^2}{2} F(\phi) F(\phi),
\end{eqnarray}

For the massive theory ($m\not=0$) this generates two types of vertices of any valence, massive vertices $\sim m^2$ and kinematic vertices $\sim p_i^2$, for external momenta $p_i$ at edge $e_i$. For any external edge $e$, set $x_e:=p_e^2-m^2$. 

We then can combine the two vertex types in a single vertex $v$ say of valence $n$, which  provides a linear combination of variables $x_e$ and $m^2$ as a Feynman rule, fixed through the coefficients $a_i$ of the field diffeomorphism.

This is our starting point in Sect.(\ref{treeamp}) for the study of tree amplitudes which a priori are rational functions in variables $x_e,m^2$.

For loop amplitudes we saw in \cite{KVdiff} already that connected $n$-point amplitudes do not vanish due to the appearance of tadpole diagrams which spoil the Hopf ideal structure. A necessary condition to regain diffeomorphism invariance is the use of a renormalization scheme  which eliminates all contributions from tadpole diagrams. These are also mathematically preferred schemes \cite{BrK}.

We define them following \cite{BrK}. We start with the vector space 
$Q_\Gamma=\mathbb{R}^{4n-10}$ for $\Gamma$ a graph with $n\geq 4$ external momenta, spanned by scalar products $q_i\cdot q_j$, subject to overall momentum conservation and four-dimensionality of spacetime. For a two-point graph, $Q_\Gamma=\mathbb{R}^{1}$, for a  three-point graph $Q_\Gamma=\mathbb{R}^{3}$, see \cite{iz,BlK1loop}. Note that for a graph with a single external leg, and hence no momentum flow through the graph (a tadpole), we have a trivial $Q_\Gamma=\mathbb{R}^{0}$.

Consider now a parametric representation of Feynman rues as given by the two Symanzik polynomials. In a kinematic renormalization scheme $R$, renormalized Feynman rules \[ \Phi_R(\Gamma): Q_\Gamma\times Q_\Gamma\to \mathbb{C}\] for a graph $\Gamma$  are of the 
form \cite{BrK}
\[
\Phi_R(\Gamma)(p,p_0)=\int_{\mathbb{P}_\Gamma}\sum_{F\in\mathcal{F}_\Gamma}(-1)^{|F|}\frac{\ln\frac{\Phi_{\Gamma/F}\psi_F+\Phi^0_F\psi_{\Gamma/F}}{\Phi^0_{\Gamma/F}\psi_F+\Phi^0_F\psi_{\Gamma/F}}}{\psi^2_{\Gamma/F}\psi^2_F} \Omega_\Gamma,
\]
where the (second Symanzik) polynomial $\Phi_h$ evaluates at $p\in Q_h$,
and $\Phi^0_h$ evaluates at $p_0\in Q_h$, for $h\subseteq\Gamma$ a sub- or co-graph of $\Gamma$ contributing in the above forest sum.

Note that by construction in such a scheme $R$, for any tadpole graph $\Gamma$, $\Phi_R(\Gamma)=0$ as necessarily $\Phi\equiv\Phi^0$, as there is no $p$ dependence in $\Phi_h$ for any $h\subseteq\Gamma$. 

Kinematic renormalization schemes include BPHZ and MOM-scheme renormalization, and quite generally any scheme where subtractions are done 
by determining renormalization conditions for the kinematics of amplitudes.
In particular, we have in such schemes 
\[
\Phi_R(\Gamma)(p_0,p_0)=0,\,\forall\Gamma.
\]
Renormalized Feynman rules for kinematic schemes have characteristic properties with respect to the renormalization group and their scale and angle dependence \cite{BrK,BlK}
as well as with respect to  their monodromies in varying external momenta over thresholds
\cite{BlKcut}.

A popular class of renormalization scheme which are not kinematic are minimmal subtraction schemes $MS,\overline{MS}$ and such.
We see below that they indeed, in contrast to kinematic schemes, fail to deliver diffeomorphism invariance.

\section{Bell polynomial identities}\label{sec bell}

The first step towards proving our results is to develop some results on Bell polynomials which will be needed for the tree-level result.

It shouldn't be surprising that Bell polynomials are important in this argument. Bell polynomials describe compositions of formal series, that is of diffeomorphisms.  Consequently a generating functions proof should be possible instead of the direct manipulations of Bell polynomials which we use, and such a proof may be able to provide more insight -- this is one of the basic insights of enumerative combinatorics and can be found in a variety of references, one recent reference is \cite{FSbook}.  However, the details are intricate in our case: the quantum field theory hides the diffeomorphism quite well and many different Bell polynomial results at different levels were needed to uncover it which would make the translation to generating functions non-trivial.  We will stick to a direct, predominantly non-generating function, approach here.

The Bell polynomials, sometimes known as \emph{partial} or \emph{incomplete} Bell polynomials, are defined as follows.

\begin{definition}Suppose $0\leq k\leq n$ are integers, then the \emph{Bell polynomials} are defined by
\[
B_{n,k} (x_1, x_2, \ldots) = \sum_{\substack{j_1+j_2+j_3+\cdots = k \\j_1+2j_2+3j_3+\cdots = n\\j_i\geq 0}}\frac{n!}{j_1!j_2!j_3!\cdots}\left(\frac{x_1}{1!}\right)^{j_1}\left(\frac{x_2}{2!}\right)^{j_2}\left(\frac{x_3}{3!}\right)^{j_3} \cdots
\]
\end{definition}
At the level of generating functions this definition becomes
\[
\exp\left(u\sum_{j=1}^{\infty}x_j\frac{t^j}{j!}\right) = \sum_{n,k\geq 0}B_{n,k}(x_1, x_2, \ldots)\frac{t^n}{n!}u^k
\]
Pulling out the coefficients of $u$ we get the composition formula for $k\geq 0$
\begin{equation}\label{eq composition 1}
\frac{1}{k!}\left(\sum_{j=1}^{\infty}x_j\frac{t^j}{j!}\right)^k = \sum_{n \geq k}B_{n,k}(x_1, x_2, \ldots)\frac{t^n}{n!}
\end{equation}

It will also be important for us that Bell polynomials count set partitions in the following sense
\[
\sum_{\substack{P_1 \cup P_2 \cup \cdots \cup P_k = \{1,\ldots, n\}\\P_i \text{ disjoint, nonempty}}} x_{|P_1|}x_{|P_2|} \cdots x_{|P_k|} = B_{n,k}(x_1, x_2, \ldots)
\]

Bell polynomial identities will be central to proving our main result.  Some of these identities were known and others can be derived using techniques in the literature.
First, note the following standard fact on shifting arguments in Bell polynomials.  The proof is classical and is a simple calculation from the definition so we will not give it.
\begin{lemma}\label{lem shift}
\[
B_{n,k}(1, x_2, x_3, x_4, \ldots) = \sum_{0 \leq j \leq k}\frac{n!}{(n-k)!j!}B_{n-k,k-j}(x_2/2, x_3/3, x_4/4, \ldots)
\]
\end{lemma}


Birmajer, Gil, and Weiner in \cite{BGWbell} give some inverse relations between Bell polynomials, the following one of which will be useful for us.

\begin{thm}[\cite{BGWbell} Theorem 15]\label{thm BGW}
Let $a, b \in \mathbb{Z}$.  Given $x_1, x_2, \ldots$, define $y_1, y_2, \ldots$ by
\[
y_n = \sum_{k=1}^{n}\binom{an+bk}{k-1}(k-1)!B_{n,k}(x_1, x_2, \ldots)
\]
Then for any $\lambda \in \mathbb{C}$, 
\[
\sum_{k=1}^n\binom{\lambda}{k-1}(k-1)!B_{n,k}(y_1, y_2, \ldots) = \sum_{k=1}^n\binom{\lambda+an+bk}{k-1}(k-1)!B_{n,k}(x_1, x_2, \ldots)
\]
\end{thm}

Next we need some identities on sums of products of Bell polynomials and their arguments.  

\begin{lemma}\label{lem classical}
Suppose $n, k>0$.
\[
B_{n,k}(x_1, x_2, \ldots) = \frac{n!}{k}\sum_{s=0}^n \frac{x_s}{s!}\frac{B_{n-s, k-1}(x_1, x_2, \ldots)}{(n-s)!}
\]
and
\[
\sum_{s=0}^{n}s \frac{x_s}{s!}\frac{B_{n-s, k-1}(x_1, x_2, \ldots)}{(n-s)!} = \frac{B_{n,k}(x_1, x_2, \ldots)}{(n-1)!}
\]
\end{lemma}

\begin{proof}
These are known.  The first is Cvijovi\'c \cite{Cbell} equation 1.4 with $1$ as $k_1$ and $k-1$ as $k_2$.  Then from Cvijovi\'c equation 2.3 with the substitution $g_n(k) = k!B_{n,k}/n!$ and $f_n = x_n/n!$ and with $k-1$ for $k$ we get
\[
\sum_{s=0}^n(sk-n)(k-1)!\frac{x_s}{s!}\frac{B_{n-s, k-1}(x_1, x_2, \ldots)}{(n-s)!}=0
\]
so we obtain the second equation
\[
\sum_{s=0}^ns\frac{x_s}{s!}\frac{B_{n-s, k-1}(x_1, x_2, \ldots)}{(n-s)!} = \frac{n}{k}\sum_{s=0}^n\frac{x_s}{s!}\frac{B_{n-s, k-1}(x_1, x_2, \ldots)}{(n-s)!}= \frac{1}{(n-1)!}B_{n,k}(x_1, x_2, \ldots)
\]
\end{proof}

The proof technique for the remaining identities follows that of Cvijovi\'c in \cite{Cbell} which he used to get the equations used above.

\begin{lemma}\label{lem FG}
Let $G(t) = \sum_{n\geq 0}g_n t^n$ and $F(t) = \sum_{n\geq 1}f_n t^n$.  Suppose $k\geq 1$ and 
\[
G(t) = F(t)^{k-1}
\]
then
\[
\sum_{i=0}^n(n(n-1) - ki(n-1))f_{i}g_{n-i} = 0
\]
\end{lemma}

\begin{proof}
$G(t) = F(t)^{k-1}$ so taking a logarithmic derivative gives
\[
G'(t)F(t) = (k-1)G(t)F'(t).
\]
Taking another derivative and rearranging gives
\[
G''(t)F(t) + (2-k)G'(t)F'(t) + (1-k)G(t)F''(t) = 0
\]
Taking coefficients gives the equation.
\end{proof}

\begin{lemma}\label{lem AB}
Let $D(t) = \sum_{n\geq 0}d_n t^n$ and $C(t) = \sum_{n\geq 1}c_n t^n$.  Suppose 
\[
D(t) = \left(\frac{1}{1-C(t)}\right)^{s+1}
\]
then, with the convention $c_0=-1$, the following identities hold
\begin{enumerate}
\item $\sum_{i=0}^n\sum_{j=0}^i(2(s+1)(i-j)j + (n-i)i)d_{n-i}c_{i-j}c_j = 0$
\item $\sum_{i=0}^n\sum_{j=0}^i(2(n-i) + (s+1)i)d_{n-i}c_{i-j}c_j = 0$
\item $\sum_{i=0}^n\sum_{j=0}^i((s+1)(i-j)j(i-2) + (n-i)(j(j-1) + (i-j)(i-j-1)))d_{n-i}c_{i-j}c_j = 0$
\item $\sum_{i=0}^n\sum_{j=0}^i((s+1)i(i-1) + (n-i)i - (s+1)(j(j-1) + (i-j)(i-j-1)))d_{n-i}c_{i-j}c_j = 0$
\end{enumerate}
\end{lemma}

\begin{proof}
$D(t) = (1-C(t))^{-(s+1)}$ so taking the logarithmic derivative we get
\begin{equation}\label{eq log der}
D'(t)(1-C(t)) = (s+1)D(t)C'(t)
\end{equation}
So
\[
2(s+1)D(t)C'(t)C'(t) = 2D'(t)(1-C(t))C'(t) = -D'(t)((1-C(t))^2)'
\]
So
\[
2(s+1)D(t)((1-C(t))')^2 + D'(t)((1-C(t))^2)' = 0
\]
Taking coefficients gives the first equation.

Returning to \eqref{eq log der}
\[
2D'(t)(1-C(t))^2 = 2(s+1)D(t)C'(t)(1-C(t)) = -(s+1)D(t)((1-C(t))^2)'
\]
so
\[
2D'(t)(1-C(t))^2 + (s+1)D(t)((1-C(t))^2)' = 0
\]
Taking coefficients gives the second equation.

Similarly, calculate
\[
(s+1)D(t)(((1-C(t))')^2)' = -2D'(t)(1-C(t))(1-C(t))''
\]
and
\[
(s+1)D(t)((1-C(t))^2)'' = -2D'(t)(1-C(t))(1-C(t))' + 2(s+1)D(t)(1-C(t))(1-C(t))''
\]
Taking coefficients gives the third and fourth equations.
\end{proof}

\begin{lemma}\label{lem bell bparts}For $1 \leq k \leq n$ integers
\[
\sum_{s=1}^{n-k+1}\frac{x_s}{s!(n-s)!}B_{n-s, k-1}(x_1, x_2, \ldots) (n(n-1)-ks(n-1)) = 0
\]
\end{lemma}

\begin{proof}
Let $f_s = x_s/s!$ and let $g_s = (k-1)!B_{s, k-1}(x_1, x_2, \ldots)/s!$.  By \eqref{eq composition 1}, apply Lemma~\ref{lem FG} and cancel the $(k-1)!$ to obtain the result.
\end{proof}

\begin{lemma}\label{lem bell prod} With the convention $x_0=-1$
\begin{enumerate}
\item \[
\sum_{i=0}^n\sum_{j=0}^i(2(s+1)(i-j)j + (n-i)i)\frac{x_{i-j}}{(i-j)!}\frac{x_j}{j!}\sum_{\ell=0}^n (s+\ell)!\frac{B_{n-i, \ell}(x_1, x_2, \ldots)}{(n-i)!} = 0 
\]
\item
\[
\sum_{i=0}^n\sum_{j=0}^i(2(n-i) + (s+1)i)\frac{x_{i-j}}{(i-j)!}\frac{x_j}{j!}\sum_{\ell=0}^n (s+\ell)!\frac{B_{n-i, \ell}(x_1, x_2, \ldots)}{(n-i)!} = 0
\]
\item
\[
\sum_{i=0}^n\sum_{j=0}^i((s+1)(i-j)j(i-2) + (n-i)(j(j-1) + (i-j)(i-j-1)))\frac{x_{i-j}}{(i-j)!}\frac{x_j}{j!}\sum_{\ell=0}^n (s+\ell)!\frac{B_{n-i, \ell}(x_1, x_2, \ldots)}{(n-i)!} = 0
\]
\item
\[
\sum_{i=0}^n\sum_{j=0}^i((s+1)i(i-1) + (n-i)i - (s+1)(j(j-1) + (i-j)(i-j-1)))\frac{x_{i-j}}{(i-j)!}\frac{x_j}{j!}\sum_{\ell=0}^n (s+\ell)!\frac{B_{n-i, \ell}(x_1, x_2, \ldots)}{(n-i)!} = 0
\]
\end{enumerate}
\end{lemma}

\begin{proof}
By \eqref{eq composition 1}, the composition formula for Bell polynomials, 
\begin{align*}
\sum_{n=0}^\infty\sum_{\ell=0}^n \frac{t^n(s+\ell)!}{n!}B_{n, \ell}(x_1, x_2, \ldots)
& = \sum_{i=0}^{\infty}\frac{(s+i)!}{i!}\sum_{j=1}^{\infty}\left(\frac{x_jt^j}{j!}\right)^i \\
& = s!\sum_{i=0}^{\infty}\binom{s+i}{i}C(t)^i \\
& = s!\sum_{i=0}^{\infty}\binom{-(s+1)}{i}(-C(t))^i \\
& = s!(1-C(t))^{-(s+1)}
\end{align*}
Where $C(t)= \sum_{i=1}^{\infty}\frac{x_i}{i!}t^i$.
So with $c_n = x_n/n!$ and $d_n = \sum_{\ell=0}^n \frac{(s+\ell)!}{s!n!}B_{n, \ell}(x_1, x_2, \ldots)$ we can apply Lemma~\ref{lem AB}.  Cancelling $s!$ gives the results.
\end{proof}

\section{Tree-level amplitudes}\label{treeamp}

\subsection{Reduction of tree-level problem to set partitions}\label{sec reduction}

The set up from \cite{KVdiff} was outlined in section~\ref{subsec prev}.  Consider in particular \eqref{L2}, the transformed part of which is repeated here
\[
L_F(\phi) = \frac{1}{2} \partial_{\mu} F(\phi) \, \partial^{\mu} F(\phi) -
 \frac{m^2}{2} F(\phi) F(\phi)
 \]
 where (see \eqref{defF})
 \[
F(\phi) = \sum_{k=0}^{\infty} a_k\phi^{k+1} = \phi + a_1 \phi^2 + a_2 \phi^3 + \ldots \qquad (\mathrm{with} \; a_0=1).
 \]
Expanding in the field, as noticed before, we obtain two types of vertex of each valence, massive vertices and kinematic vertices, and we can read off the corresponding Feynman rules from expanded Lagrangian.  Specifically, Feynman rules for the propagator are
\[
\frac{i}{p^2-m^2},
\]  
where $m$ is the mass and $p$ the momentum,
Feynman rules for the kinematic vertex of degree $n$ with momenta $p_1, p_2, \ldots, p_n$ for the incident edges are
\[
i \frac{d_{n-2}}{2}(p_1^2 + p_2^2 + \cdots + p_n^2),
\]
where
\[
d_n = n! \sum_{j=0}^{n}(j+1)(n-j+1)a_ja_{n-j},
\]
and finally, the Feynman rules for the massive vertex of degree $n$ are
\[
ic_{n-2}
\]
where
\[
c_n = -m^2\frac{(n+2)!}{2}\sum_{j=0}^na_ja_{n-j}.
\]
The tree-level $n$-point amplitude is the sum over all trees with $n$ external edges built out of these vertices.  
{}From a combinatorial perspective this is a sum over all trees with 
\begin{itemize}
  \item two different types of vertices (kinematic and massive),
  \item a momentum variable assigned to each edge (both internal and external),
  \item momentum conservation at each vertex (i.e. assign arbitrary orientations to the edges and then the sum of the momenta coming in each vertex must equal the sum of the momenta going out), 
  \item $n$ external edges (equivalently, $n$ edges which have a leaf at one end),
  \item the external edges labelled, everything else unlabelled, and
  \item $p^2=m^2$ for any momentum $p$ labelling an external edge (this is the on-shell condition, let $x_e:=p_e^2-m^2$ name the corresponding 
  off-shell variables for external edges $e$ with momentum $p_e$).
\end{itemize}
Each tree contributes the product of the Feynman rules of its internal edges and vertices (external edges contribute nothing).

As already observed, we can immediately simplify this by instead using a single vertex of each valence $n\geq 3$ which is the sum of the kinematic and massive vertices of valence $n$.  We will do this for the remainder of the section and so the $n$-point amplitude is simply the sum over all trees with $n$ external edges using only these combined vertices.
Note for later  that a different decomposition of such a combined vertex $v$ of valence $n$ is also helpful.  Specifically, $v$
can also be decomposed into $n+1$ vertices 
\[v=\sum_{j=0}^n v(j),\, v(0)=c_{n-2}+nm^2\frac{d_{n-2}}{2},\, v(j)=\frac{d_{n-2}x_j}{2},\,j>0,\] 
a fact which is useful when we study off-shell tree amplitudes.

As an example, consider the contribution to the $4$-point amplitude.  Thinking with the combined vertices we can either have a single 4-valent vertex with four external legs or one of three ways to put the four external legs enumerated $1,\ldots,4$ together using two three-point vertices connected by an internal edge.  Considering instead the fully expanded vertices which will be used in the off-shell discussion, the contribution to the $4$-point amplitude comes from 53 possible trees: For each of the three ways to put the external legs together using two three-point vertices there are four choices $v(i)$ for each three-point vertex $v$, so this gives $3\times 4 \times 4=48$ contributions. The remaining five are contributed by the five choices $w(i)$ for a four-valent vertex $w$.

\medskip

Suppose now we have an internal edge $e$ of a tree.  $e$ splits the tree into two pieces.  Thinking of drawing $e$ vertically, call these pieces the edges \emph{below} and \emph{above} $e$.  Suppose there are $n$ external edges below $e$.  Letting these edges be labelled $1, \ldots, n$ with corresponding momenta $p_1, \ldots, p_n$ then we see that edge $e$ contributes 
\[
\frac{i}{(p_1+\cdots+p_n)^2 - m^2}
\]
Now consider summing over all possible subtrees with $n$ external edges labelled by $1, \ldots, n$ below $e$.  For each such subtree apply Feynman rules to edge $e$ and the vertices and edges below $e$.  Define the result of this calculation to be $b_n$.  Let's work through some initial $b_n$ explicitly to clarify the definition.  For $b_1$, $e$ simply is external edge $1$, so this contributes $1$.  For $b_2$ the only contribution is
\[
\includegraphics{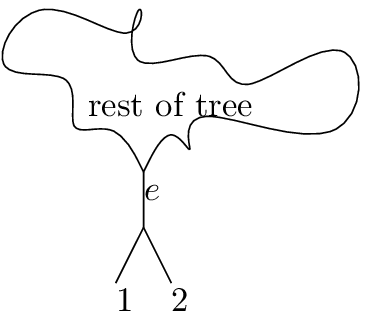}
\]
and we are only looking at the $e$ itself and the parts below $e$, namely, the contribution is
\begin{align*}
  \frac{i^2\frac{1}{2}d_1(p_1^2+p_2^1 + (p_1+p_2)^2) + i^2c_1}{(p_1+p_2)^2-m^2}
  & = -\frac{2a_1(2m^2 + (p_1+p_2)^2) - m^26a_1}{(p_1+p_2)^2-m^2} \\
  & = -\frac{2a_1((p_1+p_2)^2 - m^2)}{(p_1+p_2)^2-m^2} \\
  & = -2a_1
\end{align*}
so $b_2 = -2a_1$.
For $b_3$ there are 4 terms which contribute
\[
\includegraphics{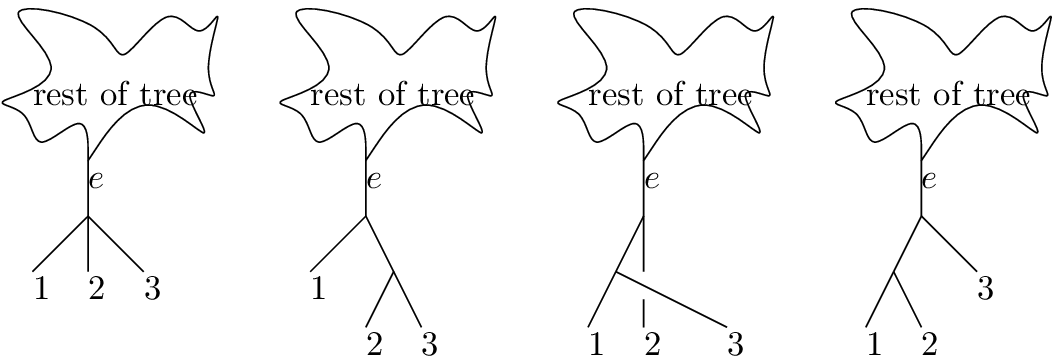}
\]
The last three of them are permutations of the labels on the same subtree shape.  Furthermore, in the last three subtrees, we don't need to recalculate the subtree consisting of an internal edge and two leaves because we already know it contributes $b_2$, so graphically we have 
\[
\includegraphics{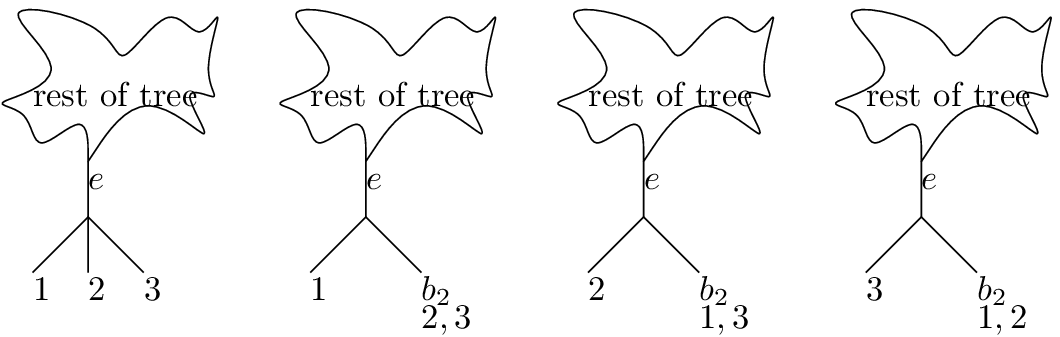}
\]
Calculating this we get
\begin{align*}
  & \frac{i^2(\frac{d_2}{2}(p_1^2+p_2^2+p_3^2+(p_1+p_2+p_3)^2) + c_2)}{(p_1+p_2+p_3)^2-m^2} + \frac{i^2(\frac{d_1}{2}(p_1^2 + (p_2+p_3)^2 + (p_1+p_2+p_3)^2) +c_1)b_2}{(p_1+p_2+p_3)^2-m^2} \\
   & + \frac{i^2(\frac{d_1}{2}(p_2^2 + (p_1+p_3)^2 + (p_1+p_2+p_3)^2) +c_1)b_2}{(p_1+p_2+p_3)^2-m^2} + \frac{i^2(\frac{d_1}{2}(p_3^2 + (p_1+p_3)^2 + (p_1+p_2+p_3)^2) +c_1)b_2}{(p_1+p_2+p_3)^2-m^2}
\end{align*}
Subbing in and completing the calculation gives $-6a_2+12a_1^2$.
In general the calculation consists of partitioning the external edges below $e$ in all possible ways, where each part of the partition contributes a factor of a previously calculated $b_i$ (see Proposition~\ref{prop bn rec} for this fact stated more formally).  Continuing to work out the first few $b_n$ explicitly we get
\begin{align*}
b_1 & = 1\\
b_2 & = -2a_1 \\
b_3 & = -6a_2 + 12a_1^2 \\
b_4 & = -24a_3 + 120a_1a_2 - 120 a_1^3 \\
b_5 & = -120a_4 + 720a_1a_3 + 360a_2^2 -2520a_1^2a_2 + 1680a_1^4
\end{align*}
Note that the $b_n$ do not depend on $m$ or the $p_i$ which is crucial, see Corollary~\ref{massindep}.  The bulk of the work for the main results consists in proving the general form for the $b_i$ given in Theorem~\ref{thm bn}.  Specifically, we will prove that
\begin{equation}\label{eq result}
b_{n+1} = \sum_{k=0}^{n} \frac{(n+k)!}{n!}B_{n,k}(-1!a_1, -2!a_2, -3!a_3, \ldots)
\end{equation}
Where the $B_{n,k}$ are the Bell polynomials (see Section~\ref{sec bell}).  The vanishing that we want is then a straightforward consequence, see Theorem~\ref{thm vanish}.

The definition of $b_n$ lets us give the following recursive expression for it.  

\begin{prop}\label{prop bn rec}
\[
b_n = - \sum_{\substack{k > 1 \\ P_1 \cup \cdots \cup P_k = \{1,\ldots, n\} \\ P_i \text{ disjoint, nonempty}}} b_{|P_1|}\cdots b_{|P_k|}\times\]
\[\times  \frac{\frac{(k-1)!}{2}\sum_{j=0}^{k-1}a_ja_{k-1-j}\left(-m^2(k+1)(k) + (j+1)(k-j)\sum_{i=1}^{k}\left(\sum_{e \in P_i}p_e\right)^2\right)}{(p_1+\cdots + p_n)^2 - m^2}.
\]
\end{prop}

\begin{proof}
Consider the edges incident to the lower vertex of $e$ other than $e$ itself.  Each defines a subtree (possibly just the edge itself as an internal edge).  
Fix a given assignment of the external edges to these subtrees.  Summing over all trees consistent with this assignment, the contribution of the each subtree is given by $b_i$ where $i$ is the number of external edges assigned to this subtree.  
All such assignments of external edges are given by set partitions of $\{1, \ldots, n\}$ with number of parts equal to the number of edges, other that $e$, incident to the lower vertex of $e$.  

Therefore $b_n$ is the sum over all possible degrees for the lower vertex of $e$ and over all possible set partitions with a compatible number of parts of the product of $b_i$ with $i$ running over the sizes of the parts of the set partition multiplied by the contribution of the lower vertex of $e$ and the contribution of $e$ itself.  This gives the statement of the proposition.
\end{proof}

We can view Proposition~\ref{prop bn rec} along with the initial condition $b_1=1$ as an alternate definition of $b_n$.  In this view, the main result is to prove that the solution to the recurrence of Proposition~\ref{prop bn rec} with initial condition $b_1=1$ is given by \eqref{eq result}.  This is a purely combinatorial problem.

From the point of view of the combinatorics $m$ is a formal variable and so in the recurrence of Proposition~\ref{prop bn rec} we can consider separately the part with $m^2$ in the numerator and the part with no $m^2$ in the numerator.  It suffices to show that the solution to these two parts separately are given by \eqref{eq result} with appropriate weights so that factoring out this common solution, the remaining coefficient of $m^2$ part and the remaining dot products in the $p_i$ occur with the correct coefficients to exactly cancel the denominator.  The remainder of this section works out the required definitions.

The $p_i$ are on-shell, so $p_i^2=m^2$.  Thus the square of any $j$ distinct $p_i$ is a sum of $jm^2$ and $j(j-1)/2$ terms of the form $2p_{i_1}\cdot p_{i_2}$.  Define, then, $b_n'$ to be given by the $m^2$ part of the recurrence of Proposition~\ref{prop bn rec}.  Specifically define
\begin{align*}
b'_n & = - \sum_{\substack{k > 1 \\ P_1 \cup \cdots \cup P_k = \{1,\ldots, n\} \\ P_i \text{ disjoint, nonempty}}} b'_{|P_1|}\cdots b'_{|P_k|} \frac{\frac{(k-1)!}{2}\sum_{j=0}^{k-1}a_ja_{k-1-j}\left(-m^2(k+1)(k) + 2nm^2(j+1)(k-j)\right)}{(n-1)m^2} \\
& = - \sum_{k = 2}^{n} B_{n,k}(b_1', b_2', \ldots) \frac{(k-1)!}{2}\sum_{j=0}^{k-1}a_j a_{k-1-j}\left(\frac{2n(j+1)(k-j) - k(k+1)}{n-1}\right)
\end{align*}

Now consider the dot product terms.  Since the entire sum is fully symmetric in the $p_i$, we do not need to keep track of which dot products appear.  We simply need to count the total number of dot product terms and by symmetry we know each possible dot product will appear equally.  Define, then, $b_n''$ to be given by the dot product part of the recurrence of Proposition~\ref{prop bn rec}.  Specifically define
\[
b''_n  = - \sum_{\substack{k > 1 \\ P_1 \cup \cdots \cup P_k = \{1,\ldots, n\} \\ P_i \text{ disjoint, nonempty}}} b''_{|P_1|}\cdots b''_{|P_k|} \frac{(k-1)!}{2}\sum_{j=0}^{k-1}a_ja_{k-1-j}\left(\sum_{i=1}^{k}\binom{|P_i|}{2} + \binom{n}{2}\right)\frac{(j+1)(k-j)}{\binom{n}{2}}
\]

\subsection{Tree-level results}

\begin{lemma}\label{lem equiv}
  Let $a_n$ and $b_n$ be sequences with $a_0=1$, $b_1=1$.
  The following are equivalent
  \begin{enumerate}
    \item $b_{n+1} = \sum_{k=0}^{n}\frac{(n+k)!}{n!}B_{n,k}(-a_1, -2!a_2, -3!a_3, \ldots)$  for $n\geq 0$
    \item $B_{m,m-n}(b_1, b_2, b_3, \ldots) = \sum_{k=0}^{n}\frac{(m-1+k)!}{(m-1-n)!n!}B_{n,k}(-a_1, -2!a_2, -3!a_3, \ldots)$  for $m > n \geq 0$
  \end{enumerate}
\end{lemma}

\begin{proof}
First note that if the second equation holds then taking the special case of $m=n+1$ we get 
\[
B_{n+1, 1}(b_1, b_2, \ldots) = \sum_{k=0}^{n}\frac{(n+k)!}{n!}B_{n,k}(-a_1, -2!a_2, -3!a_3, \ldots)
\]
and $B_{n+1,1}(b_1, b_2, \ldots) = b_{n+1}$ giving the first equation.

Now assume the first equation.  Apply the result of Birmajer, Gil, and Weiner given in Theorem~\ref{thm BGW} with $a=1$, $b=1$, $\lambda = m-n -1$, and $x_i = -i!a_i$.  Then
\begin{align*}
y_n & = \sum_{k=1}^{n}\binom{n+k}{k-1}(k-1)!B_{n,k}(-a_1, -2!a_2, -3!a_3, \ldots) \\
& = \sum_{k=1}^{n}\frac{(n+k)!}{(n+1)!}B_{n,k}(-a_1, -2!a_2, -3!a_3, \ldots) \\
& = \frac{b_{n+1}}{n+1}
\end{align*}
and
\[
\sum_{k=1}^n\binom{m-n-1}{k-1}(k-1)!B_{n,k}(y_1, y_2, \ldots) = \sum_{k=1}^n\binom{m-1+k}{k-1}(k-1)!B_{n,k}(x_1, x_2, \ldots)
\]

By Lemma~\ref{lem shift} with $b_i$ in place of $x_i$, $m$ in place of $n$, and $m-n$ in place of $k$ we get
\begin{align*}
B_{m,m-n}(b_1, b_2, b_3, b_4, \ldots) & = \sum_{\substack{m-2n \leq j\\0 \leq j\\j \leq m-n-1}}\frac{m!}{n!j!}B_{n,m-n-j}(b_2/2, b_3/3, b_4/4, \ldots) \\
& = \sum_{k=1}^n\frac{m!}{n!(m-n-k)!}B_{n,k}(b_2/2, b_3/3, b_4/4, \ldots) 
\end{align*}

Calculate
\begin{align*}
\sum_{k=1}^n\binom{m-n-1}{k-1}(k-1)!B_{n,k}(y_1, y_2, \ldots) 
& = \sum_{k=1}^n\frac{(m-n-1)!}{(m-n-k)!}B_{n,k}(b_2/2, b_3/3, \ldots) \\
& = \frac{(m-n-1)!n!}{m!}B_{m,m-n}(b_1, b_2, b_3, b_4, \ldots)
\end{align*}
So
\begin{align*}
B_{m,m-n}(b_1, b_2, b_3, b_4, \ldots) & = 
\frac{m!}{(m-n-1)!n!}\sum_{k=1}^n\binom{m-1+k}{k-1}(k-1)!B_{n,k}(x_1, x_2, \ldots) \\
& = \sum_{k=1}^n \frac{(m-1+k)!}{(m-n-1)!n!}B_{n,k}(x_1, x_2, \ldots)
\end{align*}
which is the second equation.
\end{proof}

From now on the $a_i$ are the coefficients of the original field diffeomorphism, as in Subsection~\ref{sec reduction}.  In particular $a_0=1$.

\begin{prop}\label{prop m part}
  Let $b'_n$ be defined recursively by $b'_1=1$ and 
  \begin{equation}\label{eq bprime}
b'_n = - \sum_{k = 2}^{n} B_{n,k}(b_1', b_2', \ldots) \frac{(k-1)!}{2}\sum_{j=0}^{k-1}a_j a_{k-1-j}\left(\frac{2n(j+1)(k-j) - k(k+1)}{n-1}\right)
  \end{equation}
  Then 
  \[
b'_{n+1} = \sum_{k=0}^{n} \frac{(n+k)!}{n!}B_{n,k}(-1!a_1, -2!a_2, -3!a_3, \ldots)
  \]
\end{prop}

\begin{proof}
  First, when $n=0$ we have $\sum_{k=0}^{0} \frac{(0+k)!}{0!}B_{0,k}(-1!a_1, -2!a_2, -3!a_3, \ldots) = 1 = b_1$ since $B_{0,0}=1$.  For all other values of $n$ it makes no difference if the sum in the expression for $b'_{n+1}$ starts at $1$ or at $0$ since $B_{0,k} = 0$ for $k>0$.

  Note that the missing $k=1$ term in the first sum of \eqref{eq bprime} would be exactly $b'_n$ so the recurrence \eqref{eq bprime} is equivalent to 
\begin{equation}\label{rec as eq}
\sum_{k = 1}^{n} B_{n,k}(b_1', b_2', \ldots) \frac{(k-1)!}{2}\sum_{j=0}^{k-1}a_j a_{k-1-j}\left(\frac{2n(j+1)(k-j) - k(k+1)}{n-1}\right) = 0
\end{equation}
We could prove this inductively by assuming the desired form for $b'_i$ for $i<n$ and using \eqref{eq bprime} to obtain the desired form for $b'_n$.
Equivalently we could assume the desired form for $b'_i$ for $i<n$ and then show that plugging in the desired form for $b'_n$ gives that \eqref{rec as eq} holds.  
That is, it suffices to assume
\[
b'_{i+1} = \sum_{k=0}^{i} \frac{(i+k)!}{i!}B_{i,k}(-1!a_1, -2!a_2, -3!a_3, \ldots)
\]
for $i<n$
and show that \eqref{rec as eq} holds.

So, assume 
\[
b'_{i+1} = \sum_{k=0}^{i} \frac{(i+k)!}{i!}B_{i,k}(-1!a_1, -2!a_2, -3!a_3, \ldots)
\]
for $i < n$.
By Lemma~\ref{lem equiv} we also have 
\[
B_{m,m-i}(b'_1, b'_2, b'_3, \ldots) = \sum_{\ell=0}^{i}\frac{(m-1+\ell)!}{(m-1-i)!i!}B_{i,\ell}(-a_1, -2!a_2, -3!a_3, \ldots).
\]
for $m > i$.

Taking the sum of the first and second equation of Lemma~\ref{lem bell prod} with 
$n=s$ we get
\begin{align*}
0 &= \sum_{i=0}^s\sum_{j=0}^i(2(s+1)(i-j)j + (s-i)i + 2(s-i) + (s+1)i)\frac{x_{i-j}}{(i-j)!}\frac{x_j}{j!}\sum_{\ell=0}^s (s+\ell)!\frac{B_{s-i, \ell}(x_1, x_2, \ldots)}{(s-i)!} \\
& = \sum_{i=0}^s\sum_{j=0}^i(2(s+1)(j+1)(i-j+1) - (i+1)(i+2))\frac{x_{i-j}}{(i-j)!}\frac{x_j}{j!}\sum_{\ell=0}^s (s+\ell)!\frac{B_{s-i, \ell}(x_1, x_2, \ldots)}{(s-i)!}
\end{align*}

So
\begin{align*}
& \sum_{k = 1}^{n} B_{n,k}(b_1', b_2', \ldots) \frac{(k-1)!}{2}\sum_{j=0}^{k-1}a_j a_{k-1-j}\left(\frac{2n(j+1)(k-j) - k(k+1)}{n-1}\right) \\
  & = \sum_{k = 1}^{n} \frac{(k-1)!}{2(n-1)}\sum_{\ell=0}^{n-k}\frac{(n-1+\ell)!}{(k-1)!(n-k)!}B_{n-k, \ell}(-a_1, -2!a_2, -3!a_3, \ldots)\\
  & \qquad \qquad \sum_{j=0}^{k-1}a_j a_{k-1-j}\left(2n(j+1)(k-j) - k(k+1)\right) \\
  & = \sum_{i = 0}^{s} \frac{1}{2s}\sum_{\ell=0}^{s-i}\frac{(s+\ell)!}{(s-i)!}B_{s-i, \ell}(-a_1, -2!a_2, -3!a_3, \ldots)\\
  & \qquad \qquad \sum_{j=0}^{i}a_j a_{i-j}\left(2(s+1)(j+1)(i-j+1) - (i+1)(i+2)\right)
\end{align*}
where $s=n-1$ (and $i=k-1$).  This is $0$ by the previous calculation with $x_i = -i!a_i$ and hence \eqref{rec as eq} holds.
\end{proof}

\begin{prop}\label{prop p.p part}
  Let $b''_n$ be defined recursively by $b''_1=1$ and 
  \begin{equation}\label{eq bdoubleprime}
b''_n = - \sum_{\substack{k > 1 \\ P_1 \cup \cdots \cup P_k = \{1,\ldots, n\} \\ P_i \text{ disjoint, nonempty}}} b''_{|P_1|}\cdots b''_{|P_k|} \frac{(k-1)!}{2}\sum_{j=0}^{k-1}a_ja_{k-1-j}\left(\sum_{i=1}^{k}\binom{|P_i|}{2} + \binom{n}{2}\right)\frac{(j+1)(k-j)}{\binom{n}{2}}
  \end{equation}
  for $n\geq 2$.
  Then 
  \[
b''_{n+1} = \sum_{k=0}^{n} \frac{(n+k)!}{n!}B_{n,k}(-1!a_1, -2!a_2, -3!a_3, \ldots)
  \]
\end{prop}

\begin{proof}
  As in the proof of the previous result, note that the missing $k=1$ term in the first sum of \eqref{eq bdoubleprime} would be exactly $b''_n$ so the recurrence \eqref{eq bdoubleprime} is equivalent to 
\begin{equation}\label{eq rec version 2}
\sum_{\substack{k \geq 1 \\ P_1 \cup \cdots \cup P_k = \{1,\ldots, n\} \\ P_i \text{ disjoint, nonempty}}} b''_{|P_1|}\cdots b''_{|P_k|} \frac{(k-1)!}{2}\sum_{j=0}^{k-1}a_ja_{k-1-j}\left(\sum_{i=1}^{k}\binom{|P_i|}{2} + \binom{n}{2}\right)\frac{(j+1)(k-j)}{\binom{n}{2}}
= 0
\end{equation}
for $n\geq 2$.
Next we need to understand how to deal with the $\binom{|P_i|}{2}$.  For fixed $k$ and $n\geq 2$, by the fact that Bell polynomials count set partitions and that the explicit formula for Bell polynomials, we have
\begin{align*}
  & \sum_{\substack{P_1 \cup \cdots \cup P_k = \{1,\ldots, n\} \\ P_i \text{ disjoint, nonempty}}} b''_{|P_1|}\cdots b''_{|P_k|}\sum_{i=1}^{k}\binom{|P_i|}{2} \\
  & = \sum_{s=1}^{n-k+1}\sum_{\substack{j_1+j_2+\cdots = k\\j_1+2j_2+3j_3+\cdots = n\\j_i \geq 0}} \frac{n!}{j_1!j_2!\cdots}\left(\frac{b_1''}{1!}\right)^{j_1}\left(\frac{b_2''}{2!}\right)^{j_2}\cdots \frac{j_ss(s-1)}{2} \quad \text{since $b_s''$ appears $j_s$ times}\\
  & = \sum_{s=1}^{n-k+1}\frac{b_s''s(s-1)}{2s!}\sum_{\substack{j_1+j_2+\cdots = k-1\\j_1+2j_2+3j_3+\cdots = n-s\\j_i \geq 0}} \frac{n!}{j_1!j_2!\cdots}\left(\frac{b_1''}{1!}\right)^{j_1}\left(\frac{b_2''}{2!}\right)^{j_2}\cdots \\
  & =  \sum_{s=1}^{n-k+1}\frac{b_s''s(s-1)n!}{2s!(n-s)!}B_{n-s,k-1}(b_1'', b_2'', \ldots)
\end{align*}

Using this and the first equation of Lemma~\ref{lem classical} to rearrange \eqref{eq rec version 2} we get for $n\geq 2$
\begin{align*}
& \sum_{\substack{k \geq 1 \\ P_1 \cup \cdots \cup P_k = \{1,\ldots, n\} \\ P_i \text{ disjoint, nonempty}}} b''_{|P_1|}\cdots b''_{|P_k|} \frac{(k-1)!}{2}\sum_{j=0}^{k-1}a_ja_{k-1-j}\left(\sum_{i=1}^{k}\binom{|P_i|}{2} + \binom{n}{2}\right)\frac{(j+1)(k-j)}{\binom{n}{2}}\\
& = \sum_{k=1}^{n}\frac{(k-1)!}{4}\sum_{s=1}^{n-k+1}\frac{b''_ss(s-1)n!}{s!(n-s)!}B_{n-s, k-1}(b''_1, b''_2, \ldots ) \sum_{j=0}^{k-1}a_ja_{k-1-j}\frac{(j+1)(k-j)}{\binom{n}{2}} \\
& \qquad + \sum_{k=1}^n\frac{(k-1)!}{2}B_{n,k}(b''_1, b''_2, \ldots)\sum_{j=0}^{k-1}a_ja_{k-1-j}(j+1)(k-j) \\
& = \frac{n!}{k\binom{n}{2}} \sum_{k=1}^n\sum_{j=0}^{k-1}\frac{(k-1)!}{4}a_ja_{k-1-j}(j+1)(k-j)\sum_{s=1}^{n-k+1}\frac{b''_s}{s!(n-s)!}B_{n-s, k-1}(b''_1, b''_2, \ldots)(ks(s-1)+ n(n-1))
\end{align*}
Therefore for $n\geq 2$ the recurrence \eqref{eq bdoubleprime} is equivalent to
\begin{equation}\label{p.p rec as eq}
\sum_{k=1}^n\sum_{j=0}^{k-1}\frac{(k-1)!}{2k}a_ja_{k-1-j}(j+1)(k-j)\sum_{s=1}^{n-k+1}\frac{b''_s}{s!(n-s)!}B_{n-s, k-1}(b''_1, b''_2, \ldots)(ks(s-1)+ n(n-1)) = 0
\end{equation}
As in the proof of the previous result, we can prove this inductively by assuming the desired form for $b''_i$ for $i<n$ and then showing that plugging in the desired form for $b''_n$ gives that \eqref{p.p rec as eq} holds.  
That is it suffices to assume
\[
b''_{i+1} = \sum_{k=1}^{i} \frac{(n+k)!}{n!}B_{n,k}(-1!a_1, -2!a_2, -3!a_3, \ldots)
\]
for $i<n$
and show that \eqref{p.p rec as eq} holds.

So assume
\[
b''_{i+1} = \sum_{k=1}^{i} \frac{(i+k)!}{i!}B_{i,k}(-1!a_1, -2!a_2, -3!a_3, \ldots)
\]
for $i < n$.  As before, by Lemma~\ref{lem equiv} we also have
\[
B_{m,m-i}(b''_1, b''_2, b''_3, \ldots) = \sum_{k=0}^{i}\frac{(m-1+k)!}{(m-1-i)!i!}B_{i,k}(-a_1, -2!a_2, -3!a_3, \ldots).
\]

Let's work on rewriting the right hand side of \eqref{p.p rec as eq} using the Bell polynomial identities we know.  By Lemma~\ref{lem bell bparts}
\begin{align*}
& \sum_{k=1}^n\sum_{j=0}^{k-1}\frac{(k-1)!}{2k}a_ja_{k-1-j}(j+1)(k-j)\sum_{s=1}^{n-k+1}\frac{b''_s}{s!(n-s)!}B_{n-s, k-1}(b''_1, b''_2, \ldots)(ks(s-1)+ n(n-1)) \\
& = \sum_{k=1}^n\sum_{j=0}^{k-1}\frac{(k-1)!}{2k}a_ja_{k-1-j}(j+1)(k-j)\sum_{s=1}^{n-k+1}\frac{b''_s}{s!(n-s)!}B_{n-s, k-1}(b''_1, b''_2, \ldots)ks(s+n-2) \\
& = \sum_{s=1}^{n}\frac{b''_s}{2(s-1)!}\sum_{\ell=0}^{n-s}\frac{\ell!}{(n-s)!}\sum_{j=0}^{\ell}a_ja_{\ell-j}B_{n-s, \ell}(b''_1, b''_2, \ldots)(s+n-2)(j+1)(\ell-j+1)
\end{align*}
Plug in our assumption for $B_{n-s, \ell}(b''_1, b''_2, \ldots)$ 
\begin{equation}\label{eq first step done}
\sum_{s=1}^{n}\frac{b''_s}{2(s-1)!}\sum_{\ell=0}^{n-s}\sum_{j=0}^{\ell}\sum_{k=0}^{n-s-\ell}a_ja_{\ell-j}\frac{(n-s-1+k)!}{(n-s-\ell)!}B_{n-s-\ell, k}(-a_1, -2!a_2, \ldots)\ell(s+n-2)(j+1)(\ell-j+1)
\end{equation}
Now add the sum of the first, third and fourth identities of Lemma~\ref{lem bell prod} with the following substitutions
\begin{itemize}
  \item $n-s-1$ in the place of $s$,
  \item $n-s$ in the place of $n$,
  \item $\ell$ in the place of $i$,
  \item $k$ in the place of $\ell$, and
  \item $-n!a_n$ in the place of $x_n$
\end{itemize} in order to cancel $s$ from $\ell(s+n-2)(j+1)(\ell-j+1)$ in \eqref{eq first step done}.  This gives
\begin{align*}
  & \sum_{s=1}^{n}\frac{b''_s}{2(s-1)!}\sum_{\ell=0}^{n-s}\sum_{j=0}^{\ell}\sum_{k=0}^{n-s-\ell}a_ja_{\ell-j}\frac{(n-s-1+k)!}{(n-s-\ell)!}B_{n-s-\ell, k}(-a_1, -2!a_2, \ldots)\\
  & \qquad \qquad \qquad \ell(2nj(\ell-j) + (\ell+1)(2n-\ell-2)) \\
& = 
\sum_{s=1}^{n}\frac{b''_s}{2(s-1)!}\sum_{\ell=0}^{n-s}\frac{\ell!}{(n-s)!}\sum_{j=0}^{\ell}a_ja_{\ell-j}B_{n-s, \ell}(b''_1, b''_2, \ldots)(2nj(\ell-j) + (\ell+1)(2n-\ell-2)) \\
& = \sum_{\ell=0}^{n-1}\frac{\ell!}{2}\sum_{j=0}^{\ell}a_ja_{\ell-j}(2nj(\ell-j) + (\ell+1)(2n-\ell-2))\sum_{s=1}^{n-\ell}s\frac{b''_s}{s!(n-s)!}B_{n-s, \ell}(b''_1, b''_2, \ldots) \\
& = \sum_{\ell=0}^{n-1}\frac{\ell!}{2(n-1)!}\sum_{j=0}^{\ell}a_ja_{\ell-j}(2nj(\ell-j) + (\ell+1)(2n-\ell-2)) B_{n, \ell+1}(b_1'', b_2'', \ldots)
\end{align*}
by the second equation of Lemma~\ref{lem classical}.  Then replacing $B_{n, \ell+1}$ by the sum of Bell polynomials in terms of the $a_i$ one last time we obtain
\[
\sum_{\ell=0}^{n}\frac{1}{2(n-1)!}\sum_{j=0}^{\ell}\sum_{k=0}^{n-\ell-1}a_ja_{\ell-j}\frac{(n-1-k)!}{(n-\ell-1)!}B_{n-\ell-1, k}(-a_1, -2!a_2, \ldots) (2nj(\ell-j) + (\ell+1)(2n-\ell-2))
\]
which is 0 as it is the sum of the the first two identities of Lemma~\ref{lem bell prod} with $n-1$ playing the roles of $s$ and $n$ and other substitutions as above.

Therefore \eqref{p.p rec as eq} holds proving the result.
\end{proof}

\begin{thm}\label{thm bn}
Let $b_n$ be as in Subsection~\ref{sec reduction}.  Then
\[
b_{n+1} = \sum_{k=1}^{n} \frac{(n+k)!}{n!}B_{n,k}(-1!a_1, -2!a_2, -3!a_3, \ldots)
\]
\end{thm}

\begin{proof}
The proof is by induction.  One can check directly for small values of $n$.  Assume the result holds for $i<n$.  Proposition~\ref{prop bn rec} gives a recurrence for $b_n$.  Expand all dot products so that only dot products of distinct external momenta and powers of $m^2$ remain.  By symmetry we know all these dot products appear with the same coefficient so we don't need to distinguish them.  Consider the coefficient of $m^2$ in the numerator of the right hand side of Proposition~\ref{prop bn rec}.  This gives the recurrence of Proposition~\ref{prop m part} weighted by $1/(n-1)$ which is the coefficient of $m^2$ in the denominator.  Consider the remaining parts of the numerator of the right hand side of Proposition~\ref{prop bn rec}.  These give the recurrence of Proposition~\ref{prop p.p part} weighted by $\binom{n}{2}$ which is the coefficient of the dot products in the denominator.  So factoring out the common coefficient what is left in the numerator and denominator cancels giving the desired expression for $b_n$.
\end{proof}
As a consequence we have 
\begin{cor}\label{massindep}
$b$ is independent of masses and momenta. In particular all internal propagator factor $1/x_e$, for $e$ any tree edge, cancel against numerator contributions of vertices in the sum over all trees. 
\end{cor}

\begin{thm}\label{thm vanish}
The on-shell tree-level $n$-point amplitudes of the Kreimer Velenich massive theory are $0$ for $n\geq 3$.
\end{thm}

\begin{proof}
By definition $b_{n-1}$ is the result of applying Feynman rules to the sum of all the subtrees with $n-1$ external edges below an internal edge $e$.  The result of applying Feynman rules to these same subtrees but without including the factor for the edge $e$ is 
\[
((p_1+\cdots +p_{n-1})^2 - m^2)b_{n-1}. 
\]
and by Theorem~\ref{thm bn} $b_{n-1}$ does not depend on $m$ or the $p_i$.

Consider any tree with $n$ external edges.  Let $e$ be the external edge labelled $n$.  The sum over all subtrees below $e$ with $n-1$ external edges is the same as the sum over all trees with $n$ external edges.  Edge $e$ is external now, so does not contribute.  Thus the sum we want is $((p_1+\cdots +p_{n-1})^2 - m^2)b_{n-1}$.  However, $p_1+\cdots +p_{n-1} = p_n$,  $p_n^2 = m^2$, and $b_{n-1}$ is a finite quantity, so the sum we want is $0$.
\end{proof}

\section{All loop order results}
\subsection{Symmetry Factors}
The tree-level result is a result about sums of trees, not about individual trees, so as we build up to diagrams with loops we don't have the freedom to take trees in any proportion that we like.  The first order of business for the loop result, then, is to check that diagrams are generated with the appropriate symmetry factors.  Write $\text{Sym}(G)$ for the symmetry factor of $G$.

One good way to understand symmetry factors rigorously is to view Feynman diagrams as graphs with the half-edges labelled up to isomorphism and then the labelling forgotten.  In this view the exponential generating function of these labelled objects is exactly the sum over Feynman diagrams weighted by their symmetry factor.  See Lemma 2.14 and the discussion following in \cite{Ymem} or \cite{kythesis}.  It will be helpful in the following to keep the labels through the construction and only forget them at the end.  

Note that the external edges in this view are half-edges which are not paired with another half-edge to form an internal edge.  Typically for Feynman graphs external edges are viewed as fixed and so in particular isomorphisms of the graph should not permute them.  This gives the correct symmetry factors.  Also, we will think of cutting an internal edge as breaking the half-edge-half-edge pairing which forms the edge without getting rid of the two half-edges which made it up; they simply become external edges in the pieces.

With this in mind let $G$ be a graph of this sort.  A \emph{minimal cut} or \emph{Cutkosky cut} of the graph is a set of internal edges of the graph such that if cutting these edges breaks the graph into $k$ connected components then cutting any proper subset of these edges breaks the graph into strictly fewer connected components.

{}From a graph $G$ along with a minimal cut $C$ which cuts $G$ into $k$ connected components we need to extract the following information. Let the connected components be $G_1, G_2, \ldots, G_k$, then for each $G_i$ we want to keep
\begin{itemize}
  \item The number $x_i$ of external edges of $G_i$ which were external edges of $G$,
  \item for each $i \neq j$, the number $e_{i,j}$ of external edges of $G_i$ which originally connected to $G_j$ in $G$.
\end{itemize}
If we have any set of graphs $H_1, H_2, \ldots, H_k$ where the total number of vertices in the $H_i$ equals the number of vertices of $G$ and where the external edges of each $H_i$ are partitioned into a set of size $x_i$ and sets of size $e_{i,j}$ for $i\neq j$, then we say $H_1, H_2, \ldots, H_k$ with these partitions is \emph{compatible} with the pair $G$, $C$.

For a graph $H_i$ with such a partition of its external edges we will consider an isomorphism of $H_i$ to be any bijection of the half-edges which preserves the external edges in the part of size $x_i$ and is an isomorphism of $H_i$ ignoring the partition.  Intuitively this means that the external edges which were external in the original graph are fixed but not the external edges made by the cut which is also reflected by the symmetry factor.

Given $H_1, H_2, \ldots, H_k$ with external edge partitions compatible with $G$, $C$, we can put the $H_i$ together by taking any bijection between the external edges of $H_i$ from the $e_{i,j}$ part and the external edges of $H_i$ from the $e_{j,i}$ part and using this bijection to pair the half-edges into internal edges.  Write
\[
F(H_i, H_2, \ldots, H_k)
\]
for the sum of the graphs built by running over all $\prod_{i<j}e_{i,j}!$ bijections.

\begin{prop}\label{gluetree}
Given a Feynman graph $G$ and a minimal cut $C$ consider
\[
\sum_{H_1, H_2, \ldots, H_k \text{ compatible with }G, C} \frac{1}{\text{Sym}(H_1)\text{Sym}(H_2)\cdots \text{Sym}(H_k)}F(H_i, H_2, \ldots, H_k)
\]
Then, $G$ appears in this sum weighted by exactly $\frac{1}{\text{Sym}(G)}$.
\end{prop}

\begin{proof}
Since we sum over all compatible $H_1, H_2, \ldots, H_k$, if we let $X_i$ be the sum over all $H_i$ with external edges appropriately partitioned and weighted by $\frac{1}{\text{Sym}(H_i)}$, then the homogeneous piece of 
\begin{equation}\label{eq product}
F(X_1, X_2, \ldots, X_k)
\end{equation}
with the same number of vertices as $G$ is the sum in the statement of the proposition.

Instead take the sum over all half-edges labelled $H_1, H_2, \ldots, H_k$ (up to isomorphism where the $x_i$ half-edges are fixed) and weight each one by $\frac{1}{n_i!}$ where $n_i$ is the number of half-edges of $H_i$.  Upon forgetting the labelling this will give us \eqref{eq product}.  

In the glued graphs of $X(H_1, H_2, \ldots, H_k)$ to get all possible labellings we must sum over all bijections of matching half-edges (as we do) and also consider all the ways of merging the labels from each $H_i$ into one set of labels for the result.  This is the standard product for labelled combinatorial objects and corresponds to the product of exponential generating functions.  Hence it gives all labelled graphs weighted by $\frac{1}{n!}$ where $n$ is the total number of half-edges, and upon forgetting the labellings the sums of graphs, now weighted by their symmetry, are simply multiplied, that is we get \eqref{eq product} which proves the result. 
\end{proof}

\subsection{$b_n$ off-shell}
We now progress as follows.
\begin{itemize}
\item We express amplitudes for sums of trees with a given number $j$ of off-shell external edges in terms of the dimenionless quantities $b_k$
and elementary symmetric polynomials in variables $x_e$, $x_e=q_e^2-m^2$, for off-shell edges $e$. Effectively, we can write such off-shell 
tree amplitudes in terms of internal propagators and meta-vertices provided by sums $\sum_j b_j$. 
\item Loop amplitudes are built from gluing sums of trees along $j\geq 2$ off-shell edges.
\item The Euler characteristic is used to conclude that for loop amplitudes with none or one external off-shell edge, internal edges cancel 
due to the Feynman rules so that the resulting graph is a one-vertex graph with $(|\Gamma|-1)$ (for zero external off-shell edges) or $|\Gamma|$ self-loops
(for one external off-shell edge). 
\item For zero off-shell external edges, the loop integrals vanish in any renormalization scheme. With one off-shell edge, they vanish in kinematic renormalization schemes. Accordingly, the $S$-matrix remains the unit matrix.
\item Cutkosky rules combined with dispersion relations lead to the same conclusion.
\item We outline the mechanism how an interacting field theory remains invariant under field diffeomorphisms in the context of kinematic renormalization schemes.
\item We discuss the very peculiar case of the two-point function with its two external legs off-shell, with regard to foundational properties of scattering and Haag's theorem.
\end{itemize}
\subsubsection{The tree amplitude $A^j$}
As an introductory remark, we mention that the first non-trivial coefficient $a_1$
of the field diffeomorphism provides a grading: It makes sense to regard a coefficient
$a_j,\,j>1$ as having order $a_1^j$ in $a_1$, given that  a tree on $j$ such $a_1$-vertices has order $a_1^j$
and has $j+2$ external legs, as have the vertices $d_j,c_j$.  Another way to say this is that we can think of $a_j$ as having degree $j$ and then the total degree is exactly this grading.

In this section, we investigate the behaviour of a tree-amplitude $A_n^j$  with $n$  external legs, $j$ of them off-shell, as a function of the kinematic variables 
\[ x_i:=q_i^2-m^2,\] defined by those off-shell legs $i$, $1 \leq i\leq j$. We label the off-shell legs $1,\ldots,j$, and the on-shell legs $j+1,\ldots,n$. Hence $x_i=0,\,i>j$.

Such off-shell external edges $i,\,i\leq j$ are incident to
a distinguished set of  vertices
$v_r\in V_{\mathrm{Ext}}\subseteq V_T$, $1\leq r \leq s$, with $s\leq j$ as there can be less than $j$ such vertices as two off-shell external edges might connect to the same vertex. $V_T$ is the set of all vertices of a tree $T$ contributing to $A^j_n$, similarly $E_T$ is the set of all internal edges. We set $|V_T|=:v_T$, $|E_T|=:e_T$.

The tree-amplitude $A^j=\sum_{m\geq j} A_m^j$ is a sum over contributions of all trees $T$  with $m\geq j$ external legs allowing for $j$ external off-shell and $(m-j)$ on-shell legs, and $A_m^j$ itself is defined through a sum over trees $T$
\[
A_m^j=\sum_{T\in\mathcal{T}_m} A_T,
\] 
where $A_T$ is the contribution of a tree $T$ with the given set of external legs off-shell.
Finally, $\mathcal{T}_m$ is the set of all trees with $m$ external legs. 

Note that we assume that all external momenta are in general position, so external momenta  or any partial sums of external momenta
fulfill no relations beyond momentum conservation. It follows that we can regard the set of variables $x_e$, $e\in E_T$
($e$ an internal edge) and the set of variables $x_e$, $e$ incident to $v\in V_{\mathrm{Ext}}$ ($e$ an external edge) as independent.

$A^j$ can be expanded in terms of the variables $x_i$ using elementary symmetric polynomials $E_j^i$ which defines functions  
$C_i^{(j)}=C_i^{(j)}(\{a_j\})$  such that
\begin{definition}\label{defc}
\[
A^j(\{x_i\})=:\sum_{i=0}^j E_j^i C_i^{(j)},
\]
\end{definition}
with \[
E_j^0=1,\, E_j^1=\sum_{i=1}^j x_i,\, E_j^2=\sum_{i_1<i_2}x_{i_1}x_{i_2},\, \ldots,\, E_j^j=\prod_{i=1}^j x_i,
\] 
the elementary symmetric polynomials in $j$ variables. 

We know already that the coefficient functions $C_0^{(j)}=0$, $C_1^{(j)}=b$, where
\[
b=\sum_{k=1}^\infty b_{k+1},
\]
is a formal sum omitting the constant term $b_1=1$, with $b_k$ given in Thm.(\ref{thm bn}). 
The $C_j^{(n)}$ are formal series which vanish when the diffeormorphism is trivial so that all $a_i=0$.  

To continue, we remind ourselves that we consider each tree as having one type of vertex which combines the standard and massive case.
Let $v$ be such a vertex of valence $n\geq 3$.

As announced earlier we re-expand it as
\[
v=\frac{d_{n-2}}{2} \left(\sum_{i=1}^n x_i\right)+ (c_{n-2}+n m^2\frac{d_{n-2}}{2})=:\sum_{i=0}^n v(i),
\]
with $v(0)=c_{n-2}+n m^2\frac{d_{n-2}}{2}$.
The summation runs over the $n$ edges incident to $v$
plus a constant term $v(0)\sim m^2$. As usual $x_i=q_i-m^2$.

For example for a three-valent $v$,
\[
v=\frac{d_1}{2}(x_1+x_2+x_3)+\left(\frac{3}{2}d_1m^2+c_1\right).
\]

In doing so, each vertex $v$ of valence $n$ is replaced by a sum over $n+1$ vertices $v(i)$.
If $v\in V_T$ for a tree $T$ we similarly consider $v(i)$ as an element  $v(i)\in V_T$.

For any tree $T$ we call a vertex $v(i)\in V_T$ externally marked if $i>0$ and edge $i$ is an off-shell external edge of $T$. Note that then $v(i)\in V_{\mathrm{Ext}}\subseteq V_T$. If $i$ an internal  edge of $T$, we call $i$ the marking of the vertex $v(i)$. Note that an internal edge can be at most the marking of two vertices simultaneously, as an internal edge connects two vertices.  Another way to think about the marking is as a selection of half-edges, one incident to each vertex.

An amplitude is called $k$-external if $k$ of its vertices $v_r\in V_{\mathrm{Ext}}$ are externally marked, a $0$-external amplitude is called internal. See Fig.(\ref{marking}) for an example.
\begin{figure}[!h]
\includegraphics[width=12cm]{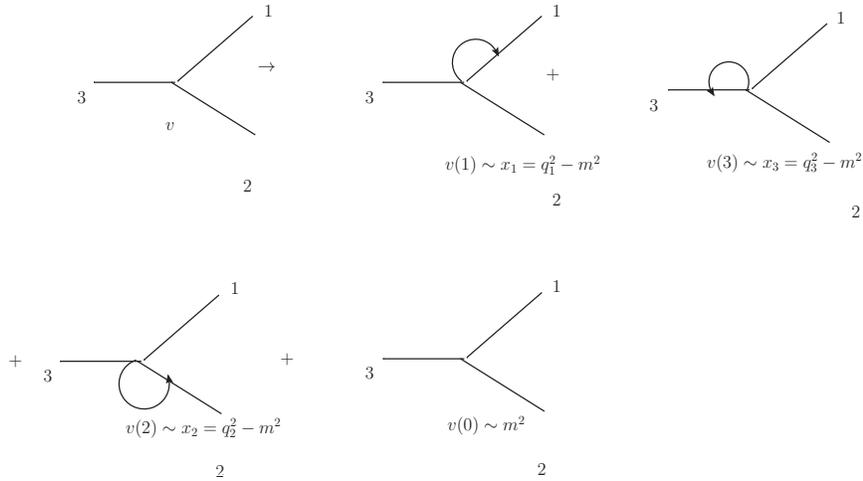}
\caption{Markings are given as a little arrow on the edge, originating from the marked vertex $v(i)$, we also give the massive vertex $v(0)$.  
We consider a vertex $v$ with three incident edges. Note that $v(i)=0$
if one of the edges $i$ is on-shell. If an edge $i$ is off-shell and internal, the vertex $v(i)$  cancels the propagator $1/x(i)$.}
\label{marking}
\end{figure}

We can organize the amplitude $C_i^{(j)}$ in terms of this decomposition of Feynman rules:
\[
C_i^{(j)}=\sum_{n=0}^i C_i^{(j)}(n)
\]
where $C_i^{(j)}(n),\,i\geq n,$ is the sum of all contributions of trees with $n$ externally marked off-shell vertices to it.
Note that the mass and momentum independence of $b$ implies that $C_1^{(j)}(0)$ is  itself independent of mass and momentum, as $C_1^{(j)}(1)$ is by definition. See also Cor.(\ref{corollaryER}).

To continue, we use that a tree with $n= |V_{\mathrm{Ext}}|\geq 3$ external edges is either a $n$-valent vertex or has an internal edge:
\begin{lemma}\label{treelemma}
Let $T^{E_n}$ be the sum of all trees with external legs labeled by the $n$-element set $E_n$.
Then, we have
\[
T^{E_n}=v+\sum_{E_n=E_{n_1}\amalg E_{n_2}} T^{E_{n_1}\cup e}\cdot e\cdot T^{E_{n_2}\cup e}, 
\]
where $v$ is a vertex of valence $n$ with its external edges labeled from $E_n$, the sum is over all partitions of $E_n$ into two disjoint
non-empty subsets $E_{n_i}$,
and $\cdot e \cdot$ implies a sum over all ways of connecting the two sums over trees by an internal edge $e$. Iterating, we get a decomposition of the sum of all trees  with with $(n-3)\geq k\geq 1$ internal edges into tree sums $T^{E_{n_i}}$, $i=1,\ldots,k+1$, and $n_1+\cdots n_{k+1}=n$ connected in all possible ways.
\end{lemma}
\begin{proof}
Definition of a labeled tree.
\end{proof}

To continue, we note that every vertex has mass dimension 2, and every internal propagator has mass dimension -2.
By the Euler characteristic, each tree amplitude has therefore dimension +2.
It follows that $C_i^{(j)}$ has dimension $-2i+2$.

Now consider $C_2^{(2)}(x_1,x_2)$. Contributions to it come from trees in which the two off-shell external edges $x_1,x_2$ connect to corresponding distinct vertices $v_1,v_2$ say (all vertices are at most linear in off-shell variables $x_i$, hence coupling two off-shell external legs $x_1,x_2$ to the same vertex only generates terms $\sim(x_1+x_2)$). Therefore there is a path $p_{12}$ between $v_1,v_2$ which contains at least one edge $e$ say,
which crucially remains unmarked.

For $C_2^{(2)}$ we have an expansion then using this intermediate off-shell propagator $1/x_e$, in particular:
\[
C_2^{(2)}(x_1,x_2)=\sum_{n=2}^\infty \sum_{n_1+n_2=n,n_i>0}b_{n_1+1}\frac{1}{x_e}b_{n_2+1}=:b\frac{1}{x_e}b,
\]
where $x_e=q_e^2-m^2$ with $q_e$ the sum of external momenta flowing into the tree sums at $v_e^+$, the vertex of $e$ closer to $v_1$, (and $v_e^-$ the vertex closer to $v_2$). Sums over trees, orientations and over all distributions of external edges are understood in this condensed notation. 

We use Lemma (\ref{treelemma}) that all trees with at least one internal propagator and $n$ external edges are obtained from connecting two sums over trees by an internal edge, and all ways of distributing $n=n_1+n_2$ external edges over them. Note that the vertices $v_e^+,v_e^-$ are internal with respect to edge $e$: edge $e$ is neither the marking for $v_e^+$ nor for $v_e^-$, as $v_e^+,v_e^-$ are 1-external emplitudes with respect to $x_1,x_2$. 
Therefore 
\[
A^2=x_1x_2 b\frac{1}{x_e} b+(x_1+x_2) b,
\]
as desired.
Note that $C^{(2)}_2$ is independent of masses and momenta,
as it factorizes into $C^{(2)}_1$ factors.

This argument continues, and $C_k^{(n)},n\geq k$, has an expansion in terms of products of $k-1$ intermediate propagators.
Crucial is the Euler characteristic, which determines for a tree $T$ that is has one more vertex than edge,  $e_T=v_T-1$.
So if $k$ vertices mark external edges, we have $k-1$ unmarked internal edges.

This determines $A^n$ completely. We set
\[
B_j:=\sum_{i=j+1}^\infty b_i,
\]
so $B_1=b$.
Now consider $k$ tree sums which are 1-external each, and connected by $k-1$ unmarked internal edges in all possible orientations. 
Regarding a 1-external tree-sum as a meta-vertex itself, of valence given by the number of internal edges incident to it,  this gives sums over meta-trees with $k$ meta-vertices of valence $\geq 1$: 
\begin{thm}\label{Cdecomp}
\[
C^{(j)}_k=\sum_{T\in\mathcal{T}_k} \prod_{v\in V_T} B_{|v|+1}\prod_{e\in E_T}\frac{1}{x_e},
\]
with $|v|$ the valency of the vertex $v$, $x_e=q_e^2-m_2$ containing the momentum flow through the internal edge $e$, and the sum is over all non-rooted trees $T\in\mathcal{T}_k$ which is  the set of trees with $k$ vertices and  with vertex set $V_T$, and the valence $|v|\geq 1$ for all vertices. 
\end{thm}
\begin{proof}
Summing over meta-trees gives the indicated $B$ factor for each vertex, and a propagator for each internal edge using Lemma (\ref{treelemma}) again.
\end{proof}
 
\subsubsection{Using the Euler characteristic}
Consider a loop amplitude with $n$ on-shell external edges.
It is a sum over all connected graphs with the indicated number of on-shell external edges and decomposes into homogeneous parts with respect to the loop number. Refine further and concentrate on those graphs which allow for a choice of $j\geq 2$ internal edges such that removing these edges decomposes the amplitude into two tree-level amplitudes $A^{j}_{n_1},\,A^{j}_{n_2}$. Summing over all $j$ reproduces the full amplitude.

Now let us go back and discuss the presence of a twice-marked edge.
As $C_1^{(j)}(1)$ is built from trees in which the number of markings equals the number of internal edges so that all internal propagators cancel 
out (directly or by the mechanism of Fig.(\ref{fig0})) it follows that it is on its own independent of masses and momenta. Hence, $C_1^{(j)}(0)$ is. After cancelling internal edges in this way in $C_1^{(j)}(0)$ there remains a single  twice-marked edge.  We can evaluate this by evaluating
\[
T^{E_n}(0):=v(0)+\sum_{E_n=E_{n_1}\amalg E_{n_2}} T^{E_{n_1}\cup e}(1)^e\cdot e\cdot T^{E_{n_2}}(1)^e,
\]
where
$T^{E_{n_i}\cup e}(1)^e$ is the set of all trees which are $1$-external with $e$ as their corresponding marked edge for both of them.
In summary we have
\begin{cor}\label{corollaryER}
\[
E_j^1C_1^{(j)}=b E_j^1=E_j^1 \left(C_1^{(j)}(0)+C_1^{(j)}(1)\right).
\]
\end{cor}
\begin{proof}
Follows from Thm.(\ref{thm bn}) and from Lem.(\ref{treelemma}), using that $v(0)$ is the only 0-external vertex.
\end{proof}
Note that $C_1^{(j)}(0)$ can be easily computed setting masses to zero using mass independence of $b_n$ and using momentum conservation.

Cor.(\ref{corollaryER}) is explained in Fig.(\ref{fig0}).
\begin{figure}[!h]
\includegraphics[width=16cm]{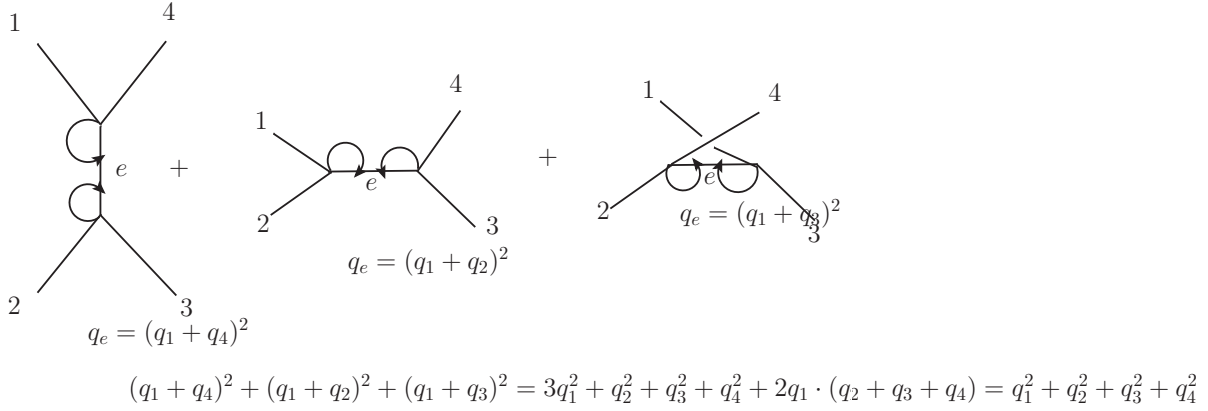}
\caption{Summing over orientations, a twice-marked edge gives a contribution proportional to the sum of its external propagators $x_i=q_i^2$
(in the massive case, massive vertices $v(0)$ guarantee the same result for $x_i=q_i^2-m^2$.}
\label{fig0}
\end{figure}
One can understand this from kinematics and momentum conservation. The momentum flow through the internal edge $e$ is given by the squared sum $(\sum_i q_i)^2$ of all external momenta $q_i$ incident to one tree sum.
Scalar products $2q_i\cdot q_j=(q_i+q_j)^2-q_i^2-q_j^2$ in that square can be replaced by sums of squares $q_i^2$ of momenta in the sum over all orientations due to momentum conservation. The internal propagators cancel, and by dimension counting, the result is linear and symmetric in off-shell variables $x_e$.  In fact, every symmetric function of variables $x_{ij}=q_i\cdot q_j$ can be replaced by suitable symmetric functions in variables given by squares $q_i^2$.

Consider two contributions $C_{i_1}^{(j)},\,C_{i_2}^{(j)}$ of such amplitudes $A^{j}_{n_1},\,A^{j}_{n_2}$.
They contain $i_1+i_2-2$ unmarked internal edges. Let $e$ be an externally marked edge of $C_{i_1}^{(j)}$
and $f$ be an externally marked edge of $C_{i_2}^{(j)}$. 

Let $v_e$ be the 1-external meta-vertex  to which edge $e$ is adjacent, similarly for $v_f$.
Glue $e,f$ together so that they form a new internal doubly marked edge $g$. Then, $(v_e,g,v_f)$ constitute an internal amplitude in an obvious manner,
as in the lhs of Fig.(\ref{fig0}).
Hence edge $g$ shrinks and the resulting sum over the two edges external to $v_e\cup v_f$ cancels either an internal unmarked edge adjacent to $v_e$ or one adjacent to $v_f$, in accordance with Fig.(\ref{fig0}) above. Summarizing, a single marked edge $e$ cancels (as $x_e/x_e=1$) and  a double-marked edge $e$ cancels itself and a neighbouring unmarked edge. It follows that the number of cancelled edges agrees with the number of vertices in total.  

As all connected graphs with loops can be obtained from gluing tree-sums in all possible ways, Prop.(\ref{gluetree}) gives us now
\begin{lemma}
i) In such a sum of graphs $G$ with $v_G$ vertices and $v_e$ edges  and $l$ loops we can cancel $v_G$ propagators if all external edges are on-shell. We are left with a single vertex with $l-1$ self-loops.\\
ii) In such a sum of graphs $G$ with $v_G$ vertices and $v_e$ edges  and $l$ loops we can cancel $v_G-1$ propagators if all external edges but one  are on-shell. We are left with a single vertex with $l$ self-loops.  
\end{lemma}
\begin{proof}
From the Euler characteristic, $e_G=v_G+l-1$ for a connected graph $G$, with $l=|G|$ its loop number.
With all external legs on-shell,  
$v_G$ internal propagators shrink leaving a single vertex. Actually,  shrinking $v_\Gamma-1$ of them leaves a rose, that is a single vertex with $l$ edges attached forming petals (self-loops). One of the petals is then still cancelled. 
If we leave one external edge off-shell, all $l$ petals remain.
\end{proof}

To continue, we use some elementary facts from kinematic renormalization.

By analytic continuation, we can consistently set
\[
\int d^Dk 1=0,
\]
which in fact is true in any renormalization scheme by analytic continuation,
and 
\[
\int d^Dk \frac{1}{k^2-m^2}=0,
\]
which is true in any kinematic renormalization scheme \cite{BrK}.

We now conclude
\begin{thm}\label{ThmLoop}
Let $A_m^{(l)}$ be a connected $l$-loop amplitude with $m$ external edges of which are at least $m-1$ are  on-shell. Let it be renormalized in kinematic renormalization conditions. Then
$A_m^{(l)}=0$.
\end{thm}
\begin{proof}
By the Euler characteristic, the amplitude (for one external edge off-shell) is proportional to  
\[
\prod_{j=1}^{l}\int\frac{d^4k_j}{k_j^2-m^2}=0,
\]
and each factor vanishes in kinematic renormalization conditions. If no external edge is off-shell, we even get an extra  factor $\int d^4 k 1 =0$ 
which vanishes under any renormalization condition even.
\end{proof} 
Note that we allow one external edge to be off-shell. With Thm.(\ref{ThmLoop}) this allows us to conclude that the propagator remains free in any scattering process. Fig.(\ref{figtwopoint}) gives the mechanism.
\begin{figure}[!h]
\includegraphics[width=10cm]{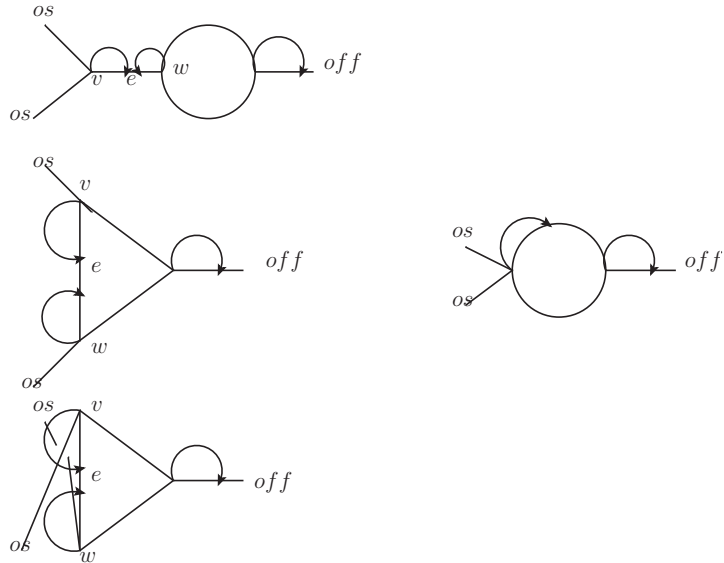}
\caption{In the three graphs on the left, a twice-marked edge $e$ connects two vertices $v,w$ in all three possible orientations. It hence shrinks, and as the only off-shell edges attached to it are internal edges of the one-loop bubble, the latter becomes a tadpole as indicated. Note that this has bearing on the LSZ formalism.}
\label{figtwopoint}
\end{figure}
This has a remarkable interpretation with regard to the LSZ formalism and asymptotic states which we consider below
in the context of field diffeomorphisms of an interacting theory.
In preparation, let us store the following lemma.
\begin{lemma}\label{lemmabanana}
For two external off-shell legs, the two-point function is supported on banana graphs. The latter are primitive elements in the Hopf algebra of renormalization.
\end{lemma}
\begin{proof}
With two off-shell external edges, we remain with a graph on two vertices. Such graphs are multiple edges between the two vertices, with possible tadpoles at either vertex. As tadpoles vanish in kinematic renormalization, we are left with pure banana graphs. The second assertion follows from the fact that every co-graph is a tadpole. Tadpole graphs form an ideal and co-ideal by which we can divide in a kinematic renormalization scheme.
See also \cite{Blat}.
\end{proof}
\subsubsection{Consistency with analyticity}
Note that the above considerations are in accord with the expectations from
the study of analytic structures of Feynman graphs \cite{BlKcut}, and in particular with the structure of iterated dispersion  derived there.

For this consider Cutkosky's theorem and dispersion relations. 
By the former, we can relate the imaginary part of an amplitude to processes with intermediate states on-shell. The latter allow us to regain the real parts from dispersion integrals over the imaginary parts.

Consider a connected one-loop amplitude with $2\leq k=k_1+k_2$ external edges on-shell, $k_1>0$ assigned to incoming states, $k_2>0$ to outgoing states.
Its imaginary part is given as a 2-particle phase space integral over two tree-level on-shell amplitudes with $k_1+1$ and $k_2+1$ on-shell external particles each. The latter amplitudes vanish, and hence does the imaginary part. So does the dispersion integral and therefore the real part, thus the full amplitude.

Note that any cut on a one-loop amplitude is a complete cut -a complete cut is set of internal edges which upon removal decomposes the graph into pieces which have vanishing first Betti number, i.e.\ no loops, in the sense of \cite{BlKcut}. Generalizing, complete cuts vanish on $l$-loop amplitudes as we are left with a phase-space integral over on-shell tree integrals. Incomplete cuts give us phase-space integrals over on-shell loop amplitudes
over $k<l$ loops. So we can use induction over the loop number using that at one loop, every cut is complete.

\section{Diffeomorphisms of an interacting theory}
Now let us add the interaction term $\frac{g}{4!}\phi^4$ to the original theory, and let us  apply the field diffeomorphisms.
Apart from the vertices constructed above, we have a new infinite set of vertices $e_n, n\geq 4$, all $\sim g$ of valency $n$, from the new 
interaction term
\[
\frac{g}{4!}\left(\phi+a_1\phi^2+a_2 \phi^3+\cdots    \right)^4=\frac{g}{4!}(\phi^4+4 a_1 \phi^5+\cdots)=\sum_{n\geq 4} e_n \phi^n.
\]

For an example, let us just look at the five-point interaction.
There is a five-point vertex $\sim 4a_1g$. But another five-point interaction comes from the connected tree diagram with a four-point interaction $\sim g$
with one of its four external legs propagating to another three-point vertex $\sim a_1$. Putting external legs on-shell, the intermediate propagator is cancelled against the Feynman rule for the three-point vertex, and we get four contributions $a_1 g$ which pair off against the contribution from the five-point vertex. 

Now consider on-shell tree sums containing one original $\phi^4$ vertex $\sim g$ and all other vertices of type $d_n,c_n$.
Shrinking internal edges between the $g$-vertex and its adjacent vertices this can be paired off with tree sums containing one vertex of type $e_n$
and all other vertices of type $d_n,c_n$. 

This pairing off eliminates the contributions of all tree sums apart from the original $g$-vertex. We conclude:
\begin{thm}
The interacting theory is diffeomorphism invariant: the $n$-point interaction
of order $g$ vanishes for $n>4$ and is $\frac{g}{4!}$ for $n=4$. 
\end{thm} 
For on-shell renormalization conditions we hence obtain identical renormalized Green functions before and after field diffeomorphisms.
 
Now assume you compute the two-point function in an interacting scalar quantum field theory, with external momentum $q$ off-shell, $x:=q^2-m^2\not=0$. Then,
\begin{lemma}\label{lemmafreetwopt}
There exists an $x$-dependent  field diffeormorphism $\{a_n=a_n(x)\}$ such that in the diffeomorphed theory the two-point self-energy function vanishes. 
\end{lemma}
\begin{proof}
We use Lem.(\ref{lemmabanana}). As all banana graphs are primitive elements in the Hopf algebra of renormalizations (all co-graphs are tadpoles),
the self-energy graphs resulting from diffeomorphisms is a series $C(\{a_n\})\ln x$. 

Let us prove  this first for the one-loop case.
\begin{figure}[!h]
\includegraphics[width=14cm]{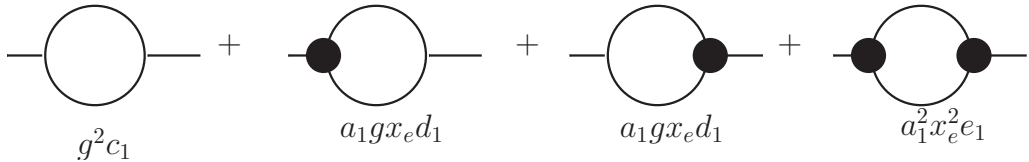}
\caption{At one-loop, the only contributing graphs involve $a_1$ and $g$ vertices.}
\label{lszone}
\end{figure}
Fig.(\ref{lszone}) shows that there is a quadratic equation for $a_1$:
\[
a_1x_e=-g \frac{d_1}{e_1}\left(1-\sqrt{1-\frac{c_1}{d_1}}\right)
\]
This patterns continues at higher loops, and there is always a quadratic equation which determines $a_k$.
Fig.(\ref{lsztwo}) shows this in the two-loop case.
\begin{figure}[!h]
\includegraphics[width=8cm]{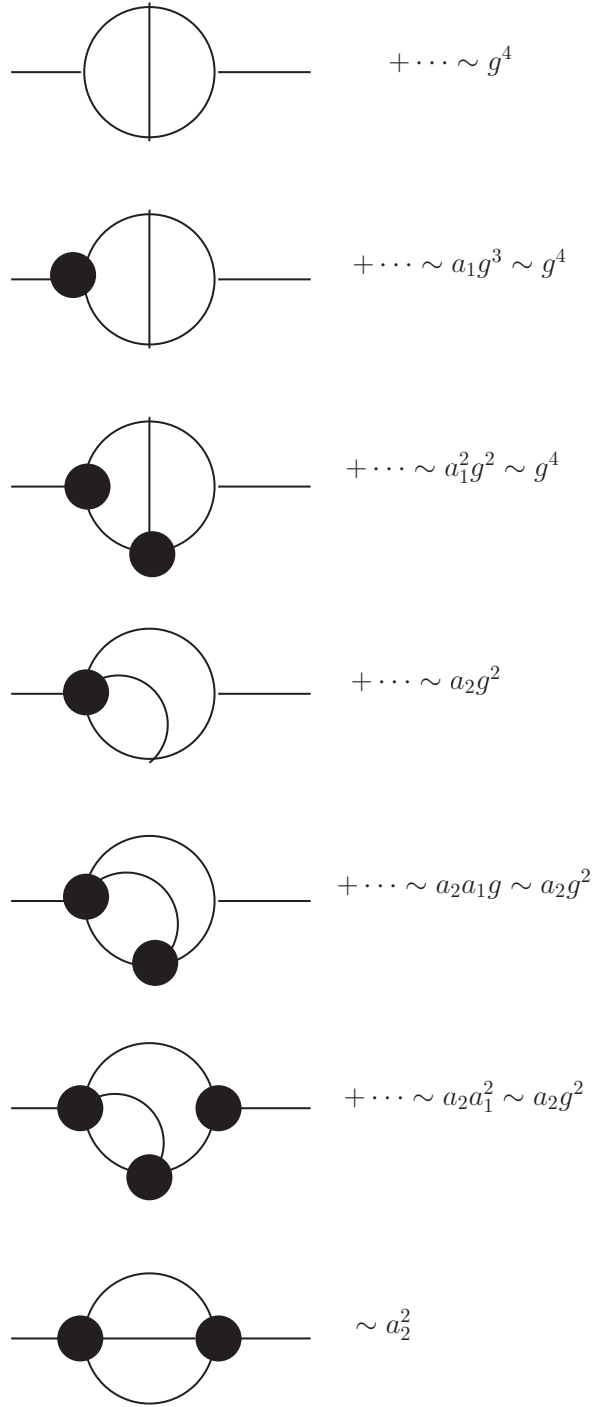}
\caption{At two-loop, all contributions can be expressed through  involve $a_2$ and $g^2$.}
\label{lsztwo}
\end{figure}
A power-counting argument shows that $a_k\sim 1/x_e^k$. Proceeding recursively, we obtain a quadratic equation for each $a_k$.
It is quadratic as each banana graph has two vertices.
\end{proof}

Now consider a quantum field theory defined by its set of edges and vertices $\mathcal{R}$ and a renormalization scheme $R$.
We call two pairs $(\mathcal{R},R)$ and $(\mathcal{R}',R')$ equivalent if they are related by a field diffeomorphism and $R,R'$ are both kinematic renormalization schemes related by a change of the renormalization point.

As two equivalent pairs give rise to identical physics, it makes sense to consider equivalence classes of such pairs.

An old problem of quantum field theory (see \cite{Duncan}, in particular section 10.5 for a clear account) is that for an interacting field theory the two-point function can not be shown to asymptotically approach the free propagator. This problem has a solution in terms of such equivalence classes.
\begin{cor}
Let $(\mathcal{R},R)$ denote an interacting quantum field theory. Then there exist a field diffeomorphism to an equivalent 
theory $(\mathcal{R}',R')$ such that in the latter the propagator is free.
In particular, computing the theory off-shell, it has a well-defined adiabatic limit
in the same equivalence class. Using the LSZ formalism to remove external propagators, the on-shell limit can then be taken in this equivalence class.
\end{cor}
\begin{proof}
Use Lem.(\ref{lemmafreetwopt}) to construct an equivalent theory which has the correct adiabatic -that is free- propagators for any off-shell $x_e\not=0$. Use the LSZ formalism to amputate connected vertex functions before taking the on-shell limit $x_e\to 0
\Leftrightarrow a_k\to\infty$.
\end{proof}

\section{Conclusion}
Let us summarize the main points of this paper.
\begin{itemize}
\item[1.] We completed the perturbative endeavour of \cite{KVdiff} and proved to all orders that a free massive field theory, after a field diffeomorphism, has no interactions
at tree level.
It indeed should not  be surprising that Bell polynomials play an important rule here because Bell polynomials can be used to describe compositions of power series.  It makes sense that applying a diffeomorphism translates to manipulating series with Bell polynomials.  However, field theory hides the original diffeomorphism very well and so the proof is far from a straightforward undoing of the original diffeomorphism, but rather an intricate  manipulation of Bell polynomials.
\item[2.] We offered two ways to extend the result to the full theory including loops, in the context of kinetic renormalization. A direct combinatorial argument featuring
the Euler characteristic delivers the result. On the other hand, the tree level result implies the vanishing of all variations of loop amplitudes. Hence, a loop amplitude could at best be a rational function of kinematic invariants, but the direct proof shows that these rational functions are absent in kinematic renormalization, as expected.
\item[3.] We gave the mechanism by which to extend these results to field diffeomorphisms of an interacting theory.
\item[4.] The problem of the adiabatic limit in an interacting quantum field theory is vexing. What has been missing so far is a clear perturbative argument how this limit could be well-defined in terms of Feynman graphs. As a first step we offer such an argument in exemplifying how  a field diffeomorphism can be constructed which diffeomorphms  the off-shell two-point propagator --- renormalized kinematical as always ---  of an interacting theory to a free
 propagator. The resulting  interacting amplitudes for connected Green functions with amputated external legs are in the same equivalence class as the adiabatically free diffeomorphed theory.
\end{itemize}
Whilst the first two points above are established in this paper, for the last two points we only outlined the basic arguments which will be expanded upon in future work.

\bibliographystyle{plain}
\bibliography{main}

\end{document}